\documentclass{article} 
\usepackage{iclr2027_conference,times}

\iclrfinalcopy

\usepackage[utf8]{inputenc}
\usepackage[T1]{fontenc}
\usepackage[british]{babel}
\usepackage{hyperref}
\usepackage{url}
\usepackage{booktabs}
\usepackage{amsmath}
\usepackage{amsfonts}
\usepackage{amssymb}
\usepackage{amsthm}
\usepackage{mdframed}
\usepackage{caption}
\usepackage{subcaption}
\usepackage{float}
\usepackage{graphicx}
\usepackage{xcolor}
\usepackage{nicefrac}
\usepackage{microtype}
\usepackage{physics}
\usepackage{tikz}
\usetikzlibrary{shapes.geometric, arrows.meta, positioning, calc, backgrounds,
                decorations.pathmorphing}

\providecommand{\internallinenumbers}{}

\newcommand{\heroscale}{0.80}

\newtheorem{proposition}{Proposition}
\newtheorem{lemma}{Lemma}
\newtheorem{corollary}{Corollary}
\newtheorem{remark}{Remark}

\newcommand{\HH}{\mathbb{H}}
\newcommand{\Enc}{\mathrm{Enc}}
\newcommand{\SU}{\mathrm{SU}}
\newcommand{\Spin}{\mathrm{Spin}}

\title{Encryptability As a Coordinate Choice:\\ Depth-One Homomorphic Federated
Learning of Quantum Neural Networks}

\author{Marcel Mordarski\thanks{Equal contribution.
  Correspondence to \texttt{marcel.mordarski25@imperial.ac.uk}.} \quad
  Nathan Mani\footnotemark[1] \quad Arshad Patel\footnotemark[1] \quad
  William Knottenbelt \quad Roberto Bondesan \\
  Department of Computing, Imperial College London \\
  London SW7 2AZ, United Kingdom}

\renewenvironment{abstract}{%
  \vskip 0.075in%
  \centerline{\large\bf Abstract}%
  \vspace{0.5ex}%
}{%
  \par%
  \vskip 1ex%
}

\begin{document}

\maketitle

\begin{abstract}
Encrypted training relies on keeping server-side updates low-degree. This constraint traditionally excludes models whose weights inhabit a compact Lie group (notably variational quantum circuits, where every trainable weight is an $\SU(2)$ rotation). Expressed in Euler angles or discrete alphabets, these updates appear transcendental, historically demanding prohibitive costs: one client--server round per gate, or upwards of $25{,}000$ operations per weight. This penalty is strictly an artefact of coordinates. In the unit-quaternion (spin) chart, group composition is exactly bilinear (degree two, with coefficients in $\{-1,0,+1\}$). Consequently, encrypted rotation updates cost one multiplicative level and federated averaging costs zero in \emph{any} levelled homomorphic scheme, completely eliminating bootstrapping. This implementation-independent algebraic property is confirmed across two cryptographic backends, introducing only $0.0$ and $-2.0\times10^{-12}$\,rad of aggregation error. Leveraging this reduction yields a non-interactive protocol for encrypted federated training of hybrid quantum--classical networks. It includes correctness proofs for aggregation and sign handling, plus a compilation lemma proving parameterised entanglers add only constant-factor overhead without altering the depth class. Empirically, a paired five-seed study confirms zero measurable utility tax ($\Delta=+9\times10^{-6}$ MSE, $p=0.92$), and a noise-budget ablation falsifies the hypothesis that encryption noise regularises. These convergence trends replicate across datasets and scale to $20$ clients. Finally, hardware validation on a $156$-qubit processor achieves $0.9918$ fidelity against a $0.99957$ unencrypted control.
\end{abstract}

\begin{figure}[H]
    \internallinenumbers
    \centering
    \scalebox{\heroscale}{%
    \begin{tikzpicture}[font=\footnotesize, line cap=round, line join=round]

    \definecolor{yarnA}{RGB}{210,108, 60}      
    \definecolor{yarnB}{RGB}{ 56,128,128}      
    \definecolor{blanketBlend}{RGB}{132,118, 96}
    \definecolor{catGrey}{RGB}{ 95, 95, 95}
    \definecolor{figGrey}{RGB}{125,125,135}
    \definecolor{needleGrey}{RGB}{ 70, 70, 80}
    \definecolor{bubbleEdge}{RGB}{160,160,170}

    \def\R{0.66}                  
    \def\Ry{0.20}                  
    \def\sphereY{1.10}            
    \def\sphereLx{-2.95}          
    \def\sphereRx{ 2.95}          

    \begin{scope}[shift={(\sphereLx, \sphereY)}]
        \draw[gray!40, thin, dashed] (\R, 0) arc (0:180:\R cm and \Ry cm);
        \draw[gray!55, thin] (\R, 0) arc (0:-180:\R cm and \Ry cm);
        \draw[black!72, thick] (0,0) circle (\R cm);
        \draw[->, thin, black!58] (0,0) -- (0, {\R+0.16}) node[above, font=\tiny] {$\ket{0}$};
        \draw[->, thin, black!58] (0,0) -- (0, {-\R-0.16}) node[below, font=\tiny] {$\ket{1}$};
        \fill[yarnA!75!black] (0,  \R) circle (1.5pt);
        \fill[yarnA!75!black] (0, -\R) circle (1.5pt);
        \fill[yarnA!75!black] ( \R, 0) circle (1.5pt);
        \fill[yarnA!75!black] (-\R, 0) circle (1.5pt);
        \fill[yarnA!75!black] (0, -\Ry) circle (1.4pt);
        \fill[yarnA!75!black, opacity=0.4] (0, \Ry) circle (1.4pt);

        \node[font=\small\bfseries, text=yarnA!70!black, anchor=south] at (0, \R + 0.40) {Pauli-OTP};
        \node[font=\tiny\itshape, text=yarnA!55!black, anchor=south] at (0, \R + 0.68) {6 discrete keys};
    \end{scope}

    \begin{scope}[shift={(\sphereRx, \sphereY)}]
        \draw[gray!40, thin, dashed] (\R, 0) arc (0:180:\R cm and \Ry cm);
        \draw[gray!55, thin] (\R, 0) arc (0:-180:\R cm and \Ry cm);
        \draw[black!72, thick] (0,0) circle (\R cm);
        \draw[->, thin, black!58] (0,0) -- (0, {\R+0.16}) node[above, font=\tiny] {$\ket{0}$};
        \draw[->, thin, black!58] (0,0) -- (0, {-\R-0.16}) node[below, font=\tiny] {$\ket{1}$};
        \foreach \a in {0,30,...,330}{
            \pgfmathsetmacro{\px}{\R*cos(\a)}
            \pgfmathsetmacro{\py}{\Ry*sin(\a)}
            \fill[yarnB!75!black, opacity=0.72] (\px, \py) circle (0.95pt);
        }
        \foreach \a in {0,40,...,320}{
            \pgfmathsetmacro{\px}{\R*0.819*cos(\a)}
            \pgfmathsetmacro{\py}{\R*0.574 + \Ry*0.819*sin(\a)}
            \fill[yarnB!75!black, opacity=0.78] (\px, \py) circle (0.95pt);
        }
        \foreach \a in {0,40,...,320}{
            \pgfmathsetmacro{\px}{\R*0.819*cos(\a)}
            \pgfmathsetmacro{\py}{-\R*0.574 + \Ry*0.819*sin(\a)}
            \fill[yarnB!75!black, opacity=0.78] (\px, \py) circle (0.95pt);
        }
        \fill[yarnB!75!black] (0,  \R) circle (1.3pt);
        \fill[yarnB!75!black] (0, -\R) circle (1.3pt);

        \node[font=\small\bfseries, text=yarnB!55!black, anchor=south] at (0, \R + 0.40) {Quaternion-OTP};
        \node[font=\tiny\itshape, text=yarnB!45!black, anchor=south] at (0, \R + 0.68) {continuous $\mathrm{SU}(2)$ keys};
    \end{scope}

    \begin{scope}[shift={(0, 1.85)}]
        \foreach \cx/\cy/\rad in {-0.78/0.10/0.32, -0.42/0.30/0.30, 0.00/0.36/0.32, 0.42/0.30/0.30, 0.78/0.10/0.32, 0.55/-0.16/0.28, 0.00/-0.20/0.30, -0.55/-0.16/0.28}{
            \fill[white] (\cx,\cy) circle (\rad);
            \draw[bubbleEdge, line width=0.45pt] (\cx,\cy) circle (\rad);
        }
        \fill[white] (-0.78,-0.05) -- (-0.85, 0.10) -- (-0.55, 0.40) -- ( 0.00, 0.50) -- ( 0.55, 0.40) -- ( 0.85, 0.10) -- ( 0.78,-0.05) -- ( 0.55,-0.30) -- ( 0.00,-0.30) -- (-0.55,-0.30) -- cycle;
        \node[font=\scriptsize, text=yarnA!70!black] at (-0.42, 0.06) {$\mathit{CNOT}$};
        \node[font=\scriptsize, text=yarnA!70!black] at (-0.42,-0.18) {$H,\;S$};
        \node[font=\small\bfseries, text=black!55] at (0, -0.05) {$+$};
        \node[font=\scriptsize, text=yarnB!50!black] at ( 0.45, 0.06) {$R_n(\theta)$};
        \node[font=\scriptsize, text=yarnB!50!black] at ( 0.45,-0.18) {$T$-gate};
    \end{scope}

    \fill[white] (0, 1.32) circle (0.10); \draw[bubbleEdge, line width=0.4pt] (0, 1.32) circle (0.10);
    \fill[white] (0, 1.05) circle (0.07); \draw[bubbleEdge, line width=0.4pt] (0, 1.05) circle (0.07);
    \fill[white] (0, 0.88) circle (0.045); \draw[bubbleEdge, line width=0.4pt] (0, 0.88) circle (0.045);

    \fill[figGrey] (0, 0.55) circle (0.22);
    \fill[figGrey!70!black, opacity=0.35] (0, 0.50) ellipse [x radius=0.20, y radius=0.06];
    \fill[figGrey] (-0.55, 0.35) .. controls (-0.50, 0.10) and (-0.40,-0.10) .. (-0.42,-0.30) -- ( 0.42,-0.30) .. controls ( 0.40,-0.10) and ( 0.50, 0.10) .. ( 0.55, 0.35) -- cycle;
    \fill[figGrey!95!black] (-0.85,-0.30) .. controls (-0.95,-0.55) and ( 0.95,-0.55) .. ( 0.85,-0.30) -- ( 0.42,-0.30) -- (-0.42,-0.30) -- cycle;
    \fill[figGrey!90!black] (-0.85,-0.55) .. controls (-1.10,-0.50) and (-1.05,-0.30) .. (-0.85,-0.30) -- cycle;
    \fill[figGrey!90!black] ( 0.85,-0.55) .. controls ( 1.10,-0.50) and ( 1.05,-0.30) .. ( 0.85,-0.30) -- cycle;

    \begin{scope}
        \fill[blanketBlend, opacity=0.90] (-0.95,-0.35) -- ( 0.95,-0.35) -- ( 1.05,-0.55) -- ( 0.85,-0.62) -- ( 0.45,-0.55) -- ( 0.05,-0.63) -- (-0.40,-0.55) -- (-0.85,-0.62) -- (-1.05,-0.55) -- cycle;
        \foreach \x in {-0.85,-0.65,...,0.85}{ \draw[yarnA!75!black, line width=0.30pt, opacity=0.55] (\x,-0.38) -- (\x,-0.55); }
        \foreach \y in {-0.40,-0.46,-0.52}{ \draw[yarnB!70!black, line width=0.32pt, opacity=0.60] (-0.92,\y) -- ( 0.92,\y); }
        \foreach \x in {-0.85,-0.60,...,0.85}{ \draw[blanketBlend!60!black, line width=0.4pt, opacity=0.7] (\x,-0.55) .. controls (\x-0.03,-0.65) and (\x+0.03,-0.65) .. (\x+0.05,-0.55); }
    \end{scope}

    \draw[figGrey, line width=3.5pt, line cap=round] (-0.55, 0.30) -- (-0.62, 0.05) -- (-0.30,-0.18);
    \fill[figGrey!85!black] (-0.30,-0.18) circle (0.07);
    \draw[figGrey, line width=3.5pt, line cap=round] ( 0.55, 0.30) -- ( 0.62, 0.05) -- ( 0.30,-0.18);
    \fill[figGrey!85!black] ( 0.30,-0.18) circle (0.07);
    \draw[needleGrey, line width=0.8pt] (-0.30,-0.18) -- ( 0.18,-0.40);
    \fill[needleGrey] ( 0.18,-0.40) circle (0.030);
    \draw[needleGrey, line width=0.8pt] ( 0.30,-0.18) -- (-0.18,-0.40);
    \fill[needleGrey] (-0.18,-0.40) circle (0.030);

    \draw[yarnA, line width=1.1pt, decoration={coil, aspect=0, segment length=2.2mm, amplitude=0.36mm}, decorate] (\sphereLx + 0.10, \sphereY - \R - 0.10) .. controls (\sphereLx + 0.30, 1.40) and (-1.40, 0.70) .. (-0.30,-0.18);
    \draw[yarnB, line width=1.1pt, decoration={coil, aspect=0, segment length=2.2mm, amplitude=0.36mm}, decorate] (\sphereRx - 0.10, \sphereY - \R - 0.10) .. controls (\sphereRx - 0.30, 1.40) and ( 1.40, 0.70) .. ( 0.30,-0.18);

    \draw[blanketBlend!70!black, line width=0.7pt, decoration={coil, aspect=0, segment length=1.4mm, amplitude=0.25mm}, decorate] ( 0.90,-0.58) .. controls ( 1.30,-0.75) and ( 1.55,-0.85) .. ( 1.85,-0.95);
    \begin{scope}[shift={(2.20,-0.95)}, scale=0.46]
        \fill[catGrey] (-0.65, 0.05) .. controls (-0.78, 0.40) and (-0.30, 0.55) .. ( 0.10, 0.50) .. controls ( 0.45, 0.45) and ( 0.55, 0.20) .. ( 0.50, 0.05) .. controls ( 0.40,-0.02) and (-0.50,-0.02) .. (-0.65, 0.05) -- cycle;
        \draw[catGrey, line width=1.7pt, line cap=round] (-0.55, 0.30) .. controls (-0.95, 0.55) and (-1.10, 1.05) .. (-0.85, 1.30);
        \draw[catGrey, line width=2.2pt, line cap=round] ( 0.20, 0.20) -- (-0.10, 0.55);
        \fill[catGrey] (-0.10, 0.55) circle (0.06);
        \fill[catGrey] ( 0.10, 0.55) circle (0.21);
        \fill[white] ( 0.03, 0.56) circle (0.029); \fill[white] ( 0.15, 0.56) circle (0.029);
        \fill[black] ( 0.018, 0.560) circle (0.018); \fill[black] ( 0.138, 0.560) circle (0.018);
    \end{scope}

    \node[font=\small\itshape, align=center, text=black!75] at (0,-1.55) {\textbf{Each yarn falls short on its own; knitting them together makes the blanket fit.}};
    \end{tikzpicture}}
\end{figure}

\begin{figure}[H]
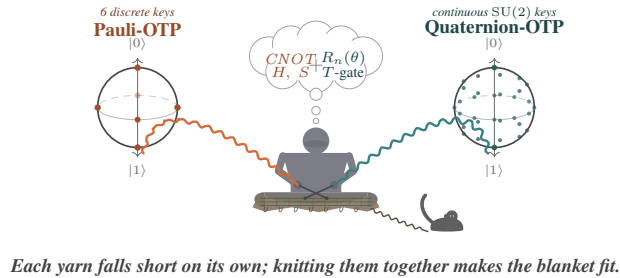

    \internallinenumbers
    \caption{\textbf{The hybrid encryption scheme as a half-finished blanket.} An exact, constant-depth compilation between disparate algebraic bases (both necessary for universal quantum computation) is proposed to encrypt federated learning. It resolves the continuous-rotation bottleneck without invoking prohibitive bootstrapping. The two yarns represent two encryption schemes for $\mathrm{SU}(2)$ used in compilation. The Pauli One-Time-Pad (orange), using discrete keys at the six Bloch lattice points, handles Clifford gates ($\mathit{CNOT}, H, S$) via a single XOR update but stalls on continuous rotations. Conversely, the Quaternion-OTP (teal), using continuous keys over the Bloch sphere, evaluates arbitrary rotations $R_n(\theta)$ via degree-2 polynomial arithmetic (Proposition~\ref{prop:degree2}) but cannot factorise across $\mathit{CNOT}$. By weaving both tracks, the protocol completes the circuit: continuous rotations are encrypted by quaternions and Cliffords by brief Pauli-track detours. The cat plays no formal role.}
    \label{fig:hero}
\end{figure}

\section{Introduction}
\label{sec:intro}

Encrypted federated learning provides a natural way to train shared models whilst protecting client privacy, and efficient implementations are rapidly emerging. Homomorphic FedAvg is practical because linear weighted averaging avoids multiplicative depth~\citep{zhang2020batchcrypt}; encrypted inference succeeds because low-degree polynomial activations eliminate the need for bootstrapping~\citep{gilad2016cryptonets} As a general principle, the reach of homomorphic encryption in ML is set by one question: \emph{is the server-side parameter update a low-degree polynomial in the parameters?} When this is the case, encryption largely amounts to bookkeeping; when it is not, general-purpose methods introduce cumbersome computational overhead.

Undermining generality, that question has been answered architecture by architecture only, and the field now
treats low-degree representability as a property of \emph{layer types}: linear
yes, ReLU no-but-approximable, attention expensive. That premise becomes harder to sustain for a growing class of models because their parameters are not free vectors in
$\mathbb{R}^d$ at all. They are elements of a compact Lie group: orthogonal and
unitary recurrent networks~\citep{arjovsky2016unitary,lezcano2019cheap},
rotation-equivariant architectures~\citep{thomas2018tensor}, quaternion-valued
networks~\citep{parcollet2020survey} and -- where the group structure is a consequence of the underlying physics -- variational quantum
circuits. Their every trainable weight is a rotation in
$\SU(2)$~\citep{benedetti2019parameterized,cerezo2021variational}. 

The final complication for group-valued setting is that the degree is coordinate-dependent: the same update may be transcendental in one chart and bilinear in another. An architecture-first analysis misses this distinction, asking about efficiency in the wrong coordinate system. For example, the Pauli
one-time pad~\citep{Broadbent_2015}, the standard encryption for quantum states,
is closed under the discrete Clifford group but not under continuous rotations:
applying $R(\theta)$ to a padded state leaves a phase error conditioned on the
secret key. 

Both accepted repairs for this phase error are expensive. Measurement-based blind
delegation~\citep{Broadbent2009Universal,fitzsimons2016privatequantumcomputationintroduction}
restores security at one client--server round \emph{per gate}, $O(T)$ in circuit
depth $T$; Solovay--Kitaev compilation~\citep{Dawson2005SK} restores it by
replacing each continuous weight with a discrete-alphabet sequence, typically
above $25{,}000$ gates per rotation~\citep{Ma_2022} where each non-Clifford element demands an active key repair.

Ultimately, these cryptographic penalties translate directly into bottlenecks. An entire model family is excluded from the encrypted federated stack the rest of ML already uses: privacy-preserving quantum federated learning is built either on plaintext updates, which gradient-inversion attacks make indefensible~\citep{zhu2019deep,geiping2020inverting}, or on interactive delegation, which does not compose with classical secure aggregation~\citep{bonawitz2017practical}. The remaining remedies force trade-offs across vulnerable dimensions -- utility for differential privacy in small-data regimes~\citep{abadi2016deep,dwork2014algorithmic}, communication for multi-party computation where round-trips bind~\citep{kairouz2021advances}. Moreover, the overheads are large enough that the protocols are never run end-to-end, leaving a literature of cost models without measurements. Ultimately, the most profound cost is conceptual: so long as encryptability is framed purely as an architectural constraint, the field remains blind to alternative solutions.

This paper deems the solution to be algebraic, as in Fig.~\ref{fig:hero}. $\SU(2)$ is isomorphic to the unit quaternions, and under that identification
group composition is quaternion multiplication, which is \emph{bilinear}:
composition is a polynomial map of total degree exactly two on $\mathbb{R}^4$
with coefficients in $\{-1,0,+1\}$ (Proposition~\ref{prop:degree2}). This shift in coordinate system perfectly satisfies the low-degree requirement for the sought efficient homomorphic encryption. Three
consequences form the backbone of this work: an encrypted rotation update consumes one multiplicative level; encrypted weighted averaging (the FedAvg primitive) consumes zero; and critically, neither needs bootstrapping. By bypassing the need for thousands of discrete gates or active key repairs, the encrypted quantum training round is \emph{shallower} in multiplicative depth than the encrypted classical MLP round it is usually compared against. Two observations make this more than a mathematical trick. The proposed encryption is scheme-agnostic, since degree two is a property of the algebra rather than of CKKS~\citep{cheon2017ckks} (verified by running the same aggregation on two independent backends in \S\ref{sec:exp-exact}). Also, the encryption generalises, since the same bilinearity holds for rotors in the even Clifford subalgebra, meaning any $\Spin(n)$-parameterised layer inherits the bound (Corollary~\ref{cor:spin}).

\paragraph{Contributions.}
\begin{itemize}\setlength{\itemsep}{0pt}
\item \textbf{An algebraic depth ledger for group-valued parameters}
  (\S\ref{sec:reduction}): Since $\SU(2)$ composition is degree-2 in spin coordinates, an encrypted rotation update consumes just one multiplicative level, and encrypted FedAvg consumes zero in \emph{any} levelled scheme (also in  $\Spin(n)$ rotors).
\item \textbf{A practical protocol and empirical evaluation}
  (\S\ref{sec:protocol}--\S\ref{sec:experiments}): Communication overhead is reduced to a single exchange per round (down from $37$--$73$). Furthermore, a paired experimental design with a matched unencrypted hardware control demonstrates no measurable utility tax, falsifying confounding regularisation effects that a single-seed run would obscure. This yields the first known public implementation of continuous-rotation quantum homomorphic encryption, validated across multiple datasets and scaled to $20$ clients.
\item \textbf{Correctness of the encrypted federated loop}
  (\S\ref{sec:correctness}--\S\ref{sec:geometry}): Aggregation correctness is established via a mask-placement lemma that separates shared from fresh randomness, alongside the derivation of the exact analytic discrepancy,
  $-\tfrac{1}{24}\sum_k w_k(\theta_k-\bar\theta)^3+O(\delta^5)$, between
  quaternion and Euclidean averaging (validated empirically to $\geqslant0.9991$).
\item \textbf{Constant-overhead entangler compilation}
  (\S\ref{sec:entanglers}): Parameterised two-qubit gates compile exactly onto
  CNOTs and single-qubit rotations, preserving the multiplicative depth class. This is verified end-to-end
  at $8.3\times10^{-3}$\,rad against a $1.25\times10^{-2}$\,rad shot-noise floor.
\end{itemize}

\section{Setting: leakage}
\label{sec:background}
To support the algebraic reduction detailed in \S\ref{sec:reduction}, the relevant quantum gates necessary for universal computation are discussed (more in the primer in Appendix~\ref{app:preliminaries}). A universal gateset splits into two categories: the discrete Clifford group (the normaliser of the Pauli group, containing gates such as $H$, $S$, and CNOT) and continuous non-Clifford operations. A hybrid quantum--classical model interleaves a parameterised quantum circuit with small classical layers, where the trainable quantum parameters are exclusively these continuous, non-Clifford single-qubit rotation angles (elements of $\SU(2)$). Gradients for these weights are evaluated via the parameter-shift rule~\citep{mitarai2018quantum,schuld2019evaluating} or simulator backpropagation. While Clifford gates are readily managed by standard quantum encryption, continuous generic rotations are not, creating the algebraic bottleneck that the quaternion encryption resolves. 

Federated training of such a model has \emph{two} distinct exposure points. \textbf{(i)~Delegated execution:} quantum hardware is remote, so a client must submit its circuit, and the rotation angles -- the weights -- cross the trust boundary before any aggregation. \textbf{(ii)~Federated aggregation:} the server that averages updates sees them, and gradient-inversion attacks reconstruct training data from exactly such transmissions~\citep{zhu2019deep,geiping2020inverting}. Both are targeted. 

An honest-but-curious server and peers and an active network adversary
on authenticated channels are assumed. The cloud QPU sees only quaternion-masked states; the
server sees ciphertexts, message sizes and the public weights $n_k/n$, holding an
evaluation-only context in which decryption raises an error; peers see the
decrypted global aggregate and nothing further by protocol
(Table~\ref{tab:threat}, Appendix~\ref{app:security}). Confidentiality of the
classical layer rests on the hardness of
(Ring-)LWE~\citep{regev2009lattices,lyubashevsky2010ideal}; the quantum layer's
rests on physics -- averaged over the pad key, the
delegated state is maximally mixed, so the guarantee is information-theoretic
and immune to unbounded compute, including a future large quantum
computer~\citep{shor1997polynomial}.

\section{Related work}
\label{sec:related}

\paragraph{Low-degree representations for encrypted ML.}
The design pattern is well established classically: CryptoNets replaces activations with low-degree polynomials to enable levelled schemes~\citep{gilad2016cryptonets}, homomorphic FedAvg exploits the linearity of averaging~\citep{zhang2020batchcrypt,cheon2017ckks}, secure aggregation relies on masking~\citep{bonawitz2017practical,shamir1979share}, and differential privacy trades utility for security~\citep{abadi2016deep,dwork2014algorithmic}. The pattern is applied here one level deeper (to the \emph{parameter group} rather than the nonlinearity), exploiting a group law that is intrinsically low-degree and requires no approximation.

\paragraph{Group-parameterised models.}
Unitary and orthogonal recurrent networks~\citep{arjovsky2016unitary, lezcano2019cheap}, rotation-equivariant architectures~\citep{thomas2018tensor}, and quaternion networks~\citep{parcollet2020survey} constrain weights to a compact group for optimisation or inductive-bias reasons. Corollary~\ref{cor:spin} demonstrates this is incidentally an encryption-friendly choice, since any $\Spin(n)$-parameterised layer inherits the depth-one bound (previously unexplored).

\paragraph{Quantum homomorphic encryption and quantum FL.}
Blind quantum computing provides information-theoretic delegation at $O(T)$ interaction~\citep{Broadbent2009Universal,fitzsimons2016privatequantumcomputationintroduction}. Pauli-OTP is non-interactive on Cliffords but stalls on continuous rotations~\citep{Broadbent_2015}, whilst Mahadev's encrypted-CNOT enables a fully classical client under LWE~\citep{mahadev2023}. Closest to the present approach, \citet{Ma_2022} introduce the quaternion representation for single-client delegated computation. The federated setting, however, shifts the requirement from isolated delegation to a continuous training loop. This is resolved here by formalising the depth ledger for that loop, establishing the correctness and geometry of aggregating group-valued updates, detailing the entangler compilation, and providing a hardware-validated implementation. Conversely, existing quantum FL frameworks either exchange plaintext updates~\citep{chen2021federated,chehimi2022quantum,Ballester2025} or rely on interactive delegation that cannot compose with classical secure aggregation~\citep{Li2025Quantum,Dutta2024ParadigmShift}; robustness to malicious clients is treated separately~\citep{ElMaouaki2025RobQFL}.

\section{Encryptability is a coordinate choice}
\label{sec:reduction}

\subsection{The degree-2 reduction and the depth ledger}
\label{sec:degree2-main}
To efficiently encrypt the continuous non-Clifford $\SU(2)$ rotations identified in \S\ref{sec:background}, a coordinate representation $q$ evaluating as a low-degree polynomial is needed. Let $\HH_1=\{q\in\mathbb{R}^4:\|q\|_2=1\}$ be the unit quaternions with the Hamilton product
\begin{equation}
(q\otimes q')_i \;=\; \sum_{j,k} M^{(i)}_{jk}\, q_j\, q'_k,
\qquad M^{(i)}_{jk}\in\{-1,0,+1\},\quad i,j,k\in\{0,1,2,3\},
\label{eq:quat-mul}
\end{equation}
and let $\Phi:\HH_1\to\SU(2)$ be the standard double cover, so that a rotation
$U(\theta,\hat u)=\cos(\theta/2)I - i\sin(\theta/2)(u_xX+u_yY+u_zZ)$ is encoded
by
\begin{equation}
q(\theta,\hat u) = \bigl(\cos\tfrac{\theta}{2},\,
-u_x\sin\tfrac{\theta}{2},\, -u_y\sin\tfrac{\theta}{2},\,
-u_z\sin\tfrac{\theta}{2}\bigr), \qquad \|q\|_2 = 1 .
\label{eq:quat-iso-app}
\end{equation}
Equation~\ref{eq:quat-mul} is the whole argument: the group law is
\emph{bilinear}, hence of total degree two, hence within reach of a levelled
homomorphic scheme at a single multiplication.

\begin{proposition}[Depth ledger for encrypted rotation update]
\label{prop:degree2}
Let $\mathcal{E}$ be any levelled homomorphic scheme supporting ciphertext
addition at depth $0$, plaintext--ciphertext multiplication at depth $1$, and
$\mathbb{R}$-linear combination with plaintext scalars at depth $0$. Let
$q\in\HH_1$ encode a rotation $\Phi(q)$ and let $V\in\SU(2)$ have quaternion
representation $r\in\HH_1$. Then \emph{(i)} $q'=q\otimes r$, representing
$\Phi(q)\Phi(r)$, is a polynomial of total degree two in $(q,r)$ with
coefficients in $\{-1,0,+1\}$; \emph{(ii)} for plaintext $r$ and encrypted $q$,
$\Enc(q')$ is computable from $\Enc(q)$ and $r$ by plaintext--ciphertext
multiplications and constant-coefficient additions alone, consuming exactly one
multiplicative level; and \emph{(iii)} the weighted mean
$\bar q=\sum_{k=1}^{K}\tfrac{n_k}{n}q^{(k)}$ of $K$ ciphertexts is computable at
depth $0$, with the renormalisation $\bar q/\|\bar q\|$ performed client-side in
plaintext after decryption and hence also at depth $0$. Consequently a federated
round of one server-side rotation update per encrypted key followed by one
weighted aggregation fits a single multiplicative level, and no
bootstrapping is required in the inner loop.
\end{proposition}

Appendix~\ref{sec:degree2} gives the proof and
Corollary~\ref{cor:fedavg-depth} the counts on unbounded round. Two remarks elucidate the value of these statement to a broader non-quantum community.

\begin{corollary}[Spin generalisation]
\label{cor:spin}
The spin group $\Spin(n)$ is the double cover of the $n$-dimensional rotation group $\mathrm{SO}(n)$, generalising the relationship between unit quaternions and $3$D rotations. Let $\Spin(n)$ act by rotors in the even subalgebra
$\mathrm{Cl}^{[0]}(n)$~\citep{lounesto2001clifford}. Because the geometric
product is bilinear, composition of rotors is a degree-2 polynomial map on the
$2^{n-1}$ even-graded coordinates, so every statement of
Proposition~\ref{prop:degree2} holds verbatim for $\Spin(n)$-parameterised
layers, with $\HH_1 \cong \Spin(3)$ being the exact case $n=3$.
\end{corollary}

The per-round accounting is as follows. The rotation update consumes one multiplicative level and requires no interaction; weighted FedAvg and the Clifford key update incur zero multiplicative cost. Meanwhile, every branch, sign, renormalisation, and recovery step remains entirely free by running client-side in plaintext. Because clients train locally in plaintext and encrypt only the resulting angles, the per-round homomorphic work reduces to a single depth-0 aggregation, \emph{whatever the ansatz contains}. Furthermore, parameter-shift gradients scale only the local plaintext quantum work, leaving the encrypted workload entirely unaffected.

\subsection{Protocol}
\label{sec:protocol}

Figure~\ref{fig:protocol} shows one federated round; the numbered walk-through
below is the definition, and every symbol in the figure is defined here rather
than in its caption.

\begin{figure}[H]
    \centering
    \resizebox{\textwidth}{!}{%
    \begin{tikzpicture}[
        >={Latex[length=2mm]},
        font=\footnotesize, line cap=round, line join=round,
        client/.style={draw, rounded corners, fill=blue!7,
                       minimum width=1.55cm, minimum height=0.55cm,
                       align=center, font=\scriptsize},
        qnnbox/.style={draw, rounded corners, fill=teal!10,
                       minimum width=1.7cm, minimum height=0.75cm,
                       align=center, font=\scriptsize},
        encbox/.style={draw, rounded corners, fill=orange!18,
                       minimum width=2.0cm, minimum height=0.95cm,
                       align=center, font=\scriptsize},
        srvbox/.style={draw, rounded corners, fill=red!8,
                       minimum width=3.5cm, minimum height=1.55cm,
                       align=center, font=\scriptsize},
        cnotbox/.style={draw, rounded corners, fill=orange!12,
                       minimum width=3.5cm, minimum height=0.85cm,
                       align=center, font=\scriptsize, dashed},
        flow/.style={->, gray!70, thick},
        encarrow/.style={->, orange!75!black, thick},
        retarrow/.style={->, red!60!black, thick},
        cnotarrow/.style={<->, orange!65!black, thick, dashed,
                         shorten <=1pt, shorten >=1pt}
    ]
    \foreach \i/\y in {1/1.6, 2/0.0, 3/-1.6}{
        \node[client] (c\i) at (0, \y)
            {Client~\i \\[-1pt] {\tiny private $\mathcal{D}_\i$}};
    }
    \node[qnnbox] (qnn) at (3.0, 0)
        {local training \\[-1pt] (plaintext) \\[-1pt]
         {\tiny angles $\theta^{(k)}$}};
    \node[encbox] (enc) at (6.2, 0)
        {encode + encrypt \\[1pt]
         {\tiny $q(\theta)=\biggl(\cos\tfrac{\theta}{2},-\hat u\sin\tfrac{\theta}{2}\biggr)$} \\[1pt]
         {\tiny shared client keys}};
    \node[srvbox] (srv) at (10.8, 0)
        {\textbf{blind server} \\[-1pt]
         {\tiny (honest-but-curious; cannot decrypt)} \\[3pt]
         {\scriptsize depth-1 rotation update:} \\[-1pt]
         {\footnotesize $\Enc(q_{\text{new}}) = \Enc(q)\cdot V^{-1}$} \\[2pt]
         {\scriptsize depth-0 homomorphic FedAvg:} \\[-1pt]
         {\footnotesize $\Enc(q_{\text{global}}) =
            \displaystyle\sum_{k=1}^{K}\!\tfrac{n_k}{n}\,\Enc\bigl(q^{(k)}\bigr)$}};
    \node[cnotbox] (cnot) at (10.8, -2.20)
        {{\tiny\textsc{delegated-execution track only}} \\[-1pt]
         {\footnotesize quaternion $\,\to\,$ Pauli-OTP $\,\to\,$ CNOT
                        $\,\to\,$ quaternion}};
    \draw[cnotarrow] (srv.south) -- (cnot.north);
    \node[font=\tiny\itshape, text=orange!55!black, anchor=west]
        at ($(srv.south)!0.5!(cnot.north) + (0.18, 0)$)
        {continuous parameters never leave the quaternion track};
    \foreach \i in {1,2,3}{ \draw[flow] (c\i.east) -- (qnn.west); }
    \draw[flow] (qnn.east) -- (enc.west);
    \draw[encarrow] (enc.east) --
        node[above, font=\tiny] {$\Enc\bigl(q^{(k)}\bigr)$} (srv.west);
    \draw[retarrow]
        (srv.south west) .. controls (7.0,-3.30) and (2.0,-3.30) ..
        node[below, font=\tiny, sloped, pos=0.55]
            {$\Enc(q_{\text{global}})$ \; (one exchange per round)}
        (c3.south);
    \begin{scope}[on background layer]
        \fill[teal!4, rounded corners=4pt] (-1.05,-1.05) rectangle (12.70, 1.05);
        \node[font=\tiny\itshape, text=teal!50!black, anchor=west]
            at (-1.0, 0.85) {quaternion track (depth $\leq 1$)};
    \end{scope}
    \end{tikzpicture}}
    \caption{One federated round. Clients train in the clear, encode each
    trainable angle as a unit quaternion (Eq.~\ref{eq:quat-iso-app}) and encrypt
    its four components under keys shared among clients. The server applies
    plaintext rotations at one multiplicative level and the weighted FedAvg sum
    at zero, never decrypting; clients decrypt, renormalise and recover angles by
    $\mathrm{atan2}$. The Pauli-OTP detour is entered only on the delegated
    execution track (Lemma~\ref{lem:entangler}).}
    \label{fig:protocol}
\end{figure}
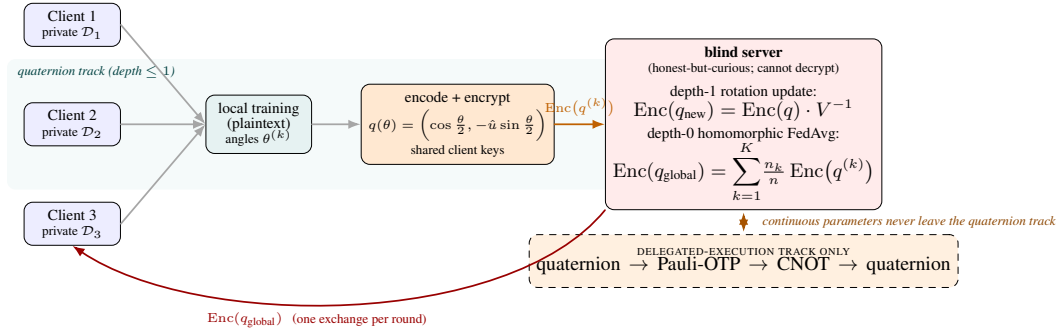

\begin{mdframed}[userdefinedwidth=0.96\textwidth, align=center]
\footnotesize
\textbf{Protocol (one federated round $t$).}
\begin{enumerate}\setlength{\itemsep}{0pt}
\item \textbf{Local training} (client, plaintext). Train on private
  $\mathcal{D}_k$ to obtain angles $\theta^{(k)}$; circuit evaluation is local,
  the delegated track of \S\ref{sec:entanglers} is used only when offloading
  to a remote QPU.
\item \textbf{Fix branch and sign, encode, encrypt} (client, depth $0$). Wrap to
  $(-\pi,\pi]$, align the hemisphere representative with the previous global
  quaternion (\S\ref{sec:geometry}), form $q\bigl(\theta^{(k)}\bigr)$ by
  Eq.~\ref{eq:quat-iso-app}, encrypt its four components under the clients'
  shared keys, upload.
\item \textbf{Aggregate} (server, depth $0$). Compute
  $\Enc(q_{\text{global}})=\sum_k\tfrac{n_k}{n}\Enc\bigl(q^{(k)}\bigr)$ by ciphertext
  addition and plaintext scaling, optionally applying a plaintext rotation
  $V^{-1}$ at depth $1$. The server holds an evaluation-only context and cannot
  decrypt.
\item \textbf{Decrypt and recover} (client, plaintext, depth $0$). Renormalise,
  recover $\hat\theta=2\,\mathrm{atan2}(-t_1,t_0)$, re-wrap, begin round $t{+}1$.
\end{enumerate}
Exactly one sequential exchange occurs per round, independent of circuit depth,
ansatz, and the number of parameter-shift evaluations used locally.
\end{mdframed}

The compatibility of the two tracks is a closure argument. Pauli-OTP is closed under Clifford conjugation but not under
continuous rotation, which is precisely the classical obstruction; the quaternion
track removes this problem because composition is degree 2; the two meet only at frame
conversions, a constant number per entangler, never sharing arithmetic on the
payload; and FedAvg, being linear, commutes with the homomorphic layer at zero
depth. Appendices~\ref{app:qotp} and~\ref{app:encrypted-qnn-training} give the
conversion machinery and the full construction.

\subsection{Correctness of the encrypted aggregation}
\label{sec:correctness}

The correctness of the reduction can be established by considering aggregation and mask-sharing.

\paragraph{Aggregation.}
The object summed by the server is the \emph{unmasked} update quaternion,
encrypted under homomorphic-layer keys shared by the clients: client $k$ encodes
each trainable angle by Eq.~\ref{eq:quat-iso-app} and encrypts the four
components, the server forms $\sum_k n_k\Enc\bigl(q^{(k)}\bigr)$ using ciphertext additions
and plaintext-integer weightings alone -- no ciphertext--ciphertext
multiplication, no decryption -- and clients decrypt, divide by $n$, renormalise
and recover $\hat\theta = 2\,\mathrm{atan2}(-t_1,t_0)$. Correctness is inherited
from linearity of the homomorphic layer exactly as in classical
HE-FedAvg~\citep{zhang2020batchcrypt}. Quaternion one-time-pad keys are a
distinct object living only on the QPU delegation channel; they are never
averaged (Appendix~\ref{app:correctness}).

\paragraph{Mask sharing.}
The following divides the labour, while Appendix~\ref{app:correctness} gives the proof and the corresponding statement
for the delegation channel.

\begin{lemma}[Mask placement]
\label{lem:mask}
Let $L_u(q)=u\otimes q$ be left multiplication by a unit quaternion $u$.
(i)~A \emph{common} left mask commutes with weighted averaging,
$\sum_k w_k L_u(q_k) = L_u\bigl(\sum_k w_k q_k\bigr)$, so it is stripped after decryption at zero depth. (ii)~A \emph{quantum-channel} mask must instead
be per-client and freshly sampled: a shared unitary $U$ preserves overlaps,
$\mathrm{Tr}[(U\rho_iU^\dagger)(U\rho_jU^\dagger)]=\mathrm{Tr}(\rho_i\rho_j)$,
leaking the pairwise geometry of client states, whereas independent Pauli keys
twirl any pair to $\tfrac{I}{d}\otimes\tfrac{I}{d}$.
\end{lemma}

\subsection{Well-definedness: branch, sign, and which mean FedAvg computes}
\label{sec:geometry}

Quaternions double-cover $\SU(2)$, meaning $q$ and $-q$ denote the same rotation, and angles wrap at $\pm\pi$, so any protocol averaging these representatives must specify both the chosen representative and the target mean. Both are resolved client-side in plaintext, incurring zero depth and zero communication. Specifically, angles are wrapped to the principal branch, $\theta\leftarrow\theta-2\pi\lfloor(\theta+\pi)/2\pi\rfloor$, and the hemisphere representative satisfying $\langle q_k,\bar q_{\text{prev}}\rangle\geqslant0$ is selected. The subsequent recovery via $\mathrm{atan2}$ is sign-exact, and renormalisation is scale-invariant. This leaves one discrepancy: the renormalised quaternion mean constitutes the chordal mean on the sphere~\citep{markley2007averaging}, whereas standard FedAvg averages angles directly.

Given a fixed axis, weights $w_k\geqslant0$ summing to one, a standard mean $\bar\theta=\sum_k w_k\theta_k$, and deviations $\delta_k=\theta_k-\bar\theta$, the angle $\hat\theta$ recovered from the renormalised weighted quaternion mean satisfies the following relation. Provided all representatives share a hemisphere,
\begin{equation}
\hat\theta - \bar\theta \;=\; -\frac{1}{24}\sum_k w_k \delta_k^3 \;+\;
O(\delta^5)
\label{eq:cubic}
\end{equation}
(Proposition~\ref{prop:cubic}, proved in Appendix~\ref{app:geometry}). This bound fails by exactly $\pi$ for configurations straddling the branch cut, a scenario excluded by construction via the aforementioned wrapping step.

The factor $\tfrac{1}{24}$ rather than the naive $\tfrac{1}{6}$ comes from the encoding: the protocol averages quaternion components at $\theta/2$. This formulation yields a discrepancy that can be computed \emph{a priori} directly from observed client disagreement, and it identifies near-antipodal disagreement as the singular configuration in which quaternion FedAvg and standard angle FedAvg diverge (measured in \S\ref{sec:exp-geometry}).

\subsection{Entanglers cost's dependence on gates and independence of depth}
\label{sec:entanglers}

A reasonable objection to any rotation-only reduction is that modern ansatzes
depend on \emph{parameterised} entanglement. It dissolves under compilation.

\begin{lemma}[Constant-overhead entangler compilation]
\label{lem:entangler}
Every controlled single-qubit rotation decomposes exactly into at most two CNOTs plus single-qubit
rotations~\citep{barenco1995elementary}, and every two-qubit unitary into at most
three~\citep{vatan2004optimal}; in particular
$R_{ZZ}(\theta)=\mathrm{CNOT}\cdot(I\otimes R_Z(\theta))\cdot\mathrm{CNOT}$. The
continuous parameter therefore always lands on a single-qubit rotation, costing
one level by Proposition~\ref{prop:degree2}(ii), while the CNOTs are fixed
Clifford elements whose key update is the depth-$0$ XOR relabelling
$(a_1,b_1,a_2,b_2)\mapsto(a_1,\,b_1\oplus b_2,\,a_1\oplus a_2,\,b_2)$, and
commutation through the pad, $R_Z(\theta)X^{a}=X^{a}R_Z((-1)^{a}\theta)$, makes
only the \emph{sign} of the angle key-conditioned. Parameterised entanglers thus
multiply gate count by a constant factor and leave the multiplicative-depth class
unchanged. (Proof in Appendix~\ref{app:correctness}.)
\end{lemma}

Much like lowering convolutions to matrix multiplication, the operation compiles exactly onto existing kernels at constant overhead, avoiding the need for custom extensions. \S\ref{sec:exp-entangler} validates this empirically.

\section{Experiments}
\label{sec:experiments}

\vspace{-0.2cm}
\subsection{Setup and claims' scope}
\label{sec:setup}

\textbf{Mechanism-level} claims -- the depth ledger,
non-interactivity, backend independence, exact entangler compilation, the
aggregation-geometry bound -- are properties of the algebra and the protocol and
are expected to generalise. \textbf{Run-level} claims -- every MSE value --
characterise a small hybrid model in one configuration, and are reported as
feasibility evidence with error bars. The task is California Housing regression~\citep{pace1997sparse},
with Wine Quality~\citep{cortez2009modeling} as an independent replication; the
model is a six-qubit, six-layer variational circuit ($36$ trainable angles, a
CNOT ring after each rotation layer) feeding a small classical head, evaluated by
exact statevector simulation with parameter-shift
gradients~\citep{Qiskit,mitarai2018quantum}. Encryption uses the LWE-based
fixed-point multi-key scheme of~\citet{Ma_2022} at $12$-bit precision unless
stated, cross-checked against CKKS over OpenFHE~\citep{badawi2022openfhe};
hardware runs use \texttt{ibm\_fez}~\citep{IBMQuantum}. Configurations, seeds and
job provenance are in Appendix~\ref{app:experiments}.
\vspace{-0.2cm}
\subsection{The depth ledger's exactness and universality}
\label{sec:exp-exact}
\vspace{-0.2cm}
\begin{table}[H]
    \centering
    \setlength{\tabcolsep}{3.5pt}
    \begin{minipage}[t]{0.485\textwidth}
        \centering
        \resizebox{\linewidth}{!}{%
        \begin{tabular}{@{}lccc@{}}
        \toprule
        \textbf{Backend} & \textbf{Single client} & \textbf{Post-FedAvg} & \textbf{Overhead} \\
        \midrule
        Multi-key (12-bit) & $1.56\times10^{-4}$ & $1.56\times10^{-4}$ & $\mathbf{0.0}$ \\
        OpenFHE CKKS & $4.76\times10^{-12}$ & $2.75\times10^{-12}$ & $\mathbf{-2.0\times10^{-12}}$ \\[2.5ex] 
        \bottomrule
        \end{tabular}}
        \vspace{1mm}
        \caption{Depth-0 aggregation exactness (float64 decode, 3 clients). Overhead isolates single-client vs.\ post-FedAvg error.}
        \label{tab:exact}
    \end{minipage}\hfill
    \begin{minipage}[t]{0.485\textwidth}
        \centering
        \resizebox{\linewidth}{!}{%
        \begin{tabular}{@{}lcc@{}}
        \toprule
        \textbf{Arm} & \textbf{Final MSE} & \textbf{Paired $\Delta$ vs.\ Plain (95\% CI)} \\
        \midrule
        Plaintext & $1.3436\pm0.72$ & --- \\
        Encrypted (12-bit) & $1.3436\pm0.72$ & $+9.0\times10^{-6}\;[-2.3,+2.5]\times10^{-4}$ \\
        Encrypted (8-bit) & $1.3442\pm0.72$ & $+6.0\times10^{-4}\;[-2.6,+3.8]\times10^{-3}$ \\
        \bottomrule
        \end{tabular}}
        \vspace{1mm}
        \caption{Utility under encryption (5 paired seeds, 3 clients). Maximum round-0 loss gap is exactly $0.00$ across all seeds.}
        \label{tab:utility}
    \end{minipage}
    \vspace{-3mm}
\end{table}

Proposition~\ref{prop:degree2}(iii) predicts that weighted aggregation is not merely cheap but \emph{error-free relative to per-client encoding}, in any levelled scheme. This prediction is evaluated via a single measurement -- three clients, $36$ angles each: encode, aggregate homomorphically, decrypt, compare -- executed on two cryptographically distinct backends, decoding in float64 so no cast floor masks the result. Aggregation overhead is $0.0$,rad on the fixed-point multi-key stack (single-client and post-FedAvg angle MAE both $1.5632\times10^{-4}$) and $-2.0\times10^{-12}$,rad on OpenFHE CKKS, with the latter representing independent rounding errors averaging down rather than a systematic effect (Table~\ref{tab:exact}). The two stacks differ by eight orders of magnitude in their inherent encoding error yet agree exactly on the algebraic depth-1 and depth-0 properties. Finally, the same run demonstrates that the noise budget functions as a dial -- each precision bit roughly halves the angle error, spanning a $60\times$ range (Table~\ref{tab:precision}) -- establishing the $12$-bit operating point at $1.52\times10^{-4}$\,rad as a deliberate set-point.
\vspace{-0.2cm}
\subsection{Utility as a function of encryption and noise}
\label{sec:exp-parity}

The practical question is simpler: what effect does switching encryption on have on the loss? Answering this requires a paired design, as the between-seed spread of the final MSE is $\approx0.72$ -- three to five orders of magnitude larger than any plausible encryption effect. Data loaders are reseeded and rebuilt immediately before each arm, isolating the encryption round trip as the \emph{only} variable. In experiments, arms starting identically must yield equal round-0 loss, and the maximum round-0 gap across five seeds is exactly $0$. The result is parity -- plaintext $1.3436\pm0.7228$ against encrypted $1.3436\pm0.7229$, with a mean paired difference of $+9\times10^{-6}$ MSE ($95\%$ CI $[-2.3,+2.5]\times10^{-4}$, $p=0.923$ paired $t$, and $p=0.625$ Wilcoxon; Table~\ref{tab:utility}).

A single-seed comparison in this configuration would have shown the encrypted arm ahead, and the available explanation -- that bounded approximation noise acts as a stochastic regulariser on a rugged variational landscape~\citep{mcclean2018barren} -- is superficially compelling. However, it makes a falsifiable prediction: increased noise should yield greater improvements. To test this, the encrypted arm was re-run at $8$-bit precision (a $16\times$ larger quantisation step, identical seeds). The noisier arm offers no improvement ($\Delta=+0.0006$ MSE, $p=0.636$). Thus, encryption incurs no utility tax.

\subsection{The aggregation-geometry bound}
\label{sec:exp-geometry}

Proposition~\ref{prop:cubic} is evaluated by scattering five clients' $36$ angles around a base vector at dispersions $\sigma\in\{0.05,0.2,0.5,1.0\}$\,rad and comparing the realised chordal-versus-Euclidean difference through the secure stack against the closed-form prediction. The cubic term tracks the measured bias with correlations of $1.0000$, $1.0000$, $0.9999$ and $0.9991$. At the operating dispersion, the $2.50\times10^{-6}$\,rad bias closely matches the predicted $2.49\times10^{-6}$ and sits $47\times$ below the concurrent encryption noise floor, ensuring the choice of mean alters no reported figure (Figure~\ref{fig:results}c). An adversarial double-cover configuration is also constructed (client groups straddling the branch cut, $\approx2\pi$ apart), where the naive arithmetic mean errs by $3.1416$\,rad against $0.0000$ for the exact chordal mean. This defines the strict boundary of the bound, but it never arises in trained runs, as the operations in \S\ref{sec:geometry} structurally exclude it at zero depth and zero communication by re-wrapping into $(-\pi,\pi]$ every round.
\vspace{-0.2cm}
\subsection{Parameterised entanglers}
\label{sec:exp-entangler}

Lemma~\ref{lem:entangler} predicts $R_{ZZ}(\theta)$ traverses the encrypted pipeline as two fixed CNOTs and one rotation, adding no rotation-track depth. This exact circuit was executed end-to-end for $20$ angles in $[0,\pi]$ at $4096$ shots each. The mean recovered-angle error is $8.3\times10^{-3}$\,rad (maximum $2.1\times10^{-2}$), evaluated against the estimator's intrinsic floor: a $Z$-basis estimate from $N$ shots yields an expected pure-noise MAE of $1.25\times10^{-2}$\,rad ($\sigma=1.6\times10^{-2}$\,rad). With the measured mean strictly \emph{below} this floor and the worst trial reaching only $1.3\sigma$, the encrypted entangler introduces no resolvable angle error. Furthermore, every trial consumed exactly two CNOTs independently of $\theta$, and the second qubit returned to $\ket{0}$ with probability $1.000$---a direct verification of the XOR key relabelling, as any update failure would leave residual entanglement.

Per trial, encryption requires $6$\,ms and decryption $51$\,ms, whilst homomorphic CNOT key updates consume $\approx1.5$\,h per gate. Since this update is algorithmically $O(1)$ and angle-independent (two zero-depth XOR operations on pad bits), this disparity stems entirely from implementation overhead, not complexity. Boolean gates are currently evaluated within a single-threaded fixed-point reference chosen for auditability; migrating to a Boolean-native TFHE scheme~\citep{chillotti2020tfhe} resolves this. Crucially, this cost vanishes in the federated setting, where local plaintext training reduces per-round homomorphic work to a single depth-0 aggregation; it applies solely to delegated encrypted execution (Appendix~\ref{exp:e4}). The run thus proves exactness and constant overhead, but not throughput.
\vspace{-0.2cm}
\subsection{Hardware}
\label{sec:hardware}

An encryption layer contributing more than the device's native gate error would
cap achievable accuracy regardless of training budget. The encrypted round trip
on \texttt{ibm\_fez} attains state fidelity $F=0.9918$. To \emph{isolate} the
encryption's contribution,a matched unencrypted
control was run -- $100$ angles under the same seeded sampling scheme, $4096$ shots
each, same device, calibration snapshot alongside -- giving $F=0.99957$
(s.d.\ $0.00107$). The $0.0078$ difference is comparable to the device's median
readout error at the time ($0.0083$), i.e.\ ordinary device variation, while the encryption arithmetic is separately bounded at $1.55\times10^{-4}$\,rad by the
homomorphic-layer-only round trip. The system is hardware-limited, not
encryption-limited.
\vspace{-0.2cm}
\subsection{Independence from the dataset and client count}
\label{sec:exp-scale}

The evaluation is extended in two directions (five seeds per point; tables in Appendix~\ref{exp:e6}). Scaling from $3$ to $20$ clients increases the final MSE slightly: from $1.3035$ to $1.3763$ on California Housing, and from $1.4214$ to $1.5246$ on Wine Quality. This degradation is mild and monotonic, showing no aggregation blow-up despite a $6.7\times$ jump in participants. Furthermore, training out to $50$ rounds improves the mean on both datasets ($1.2285\to1.0831$; $1.3148\to1.1275$) and narrows the seed spread ($0.54\to0.37$; $0.20\to0.09$). This behaviour confirms standard federated convergence, entirely unhindered by the encryption layer. Finally, although Wine Quality differs in feature count, scale, and target, it reproduces every qualitative trend. This includes a precision--bandwidth relation that is numerically identical across both datasets, confirming it as a property of the protocol itself.
\vspace{-0.2cm}
\subsection{Memory cost}
\label{sec:comm}

Non-interactivity is bought with bandwidth, and the presentation states
where the trade stops being beneficial. Per client-round the protocol moves $22.67$\,MB up
and the same down in \emph{one} sequential exchange; interactive baselines move
about $0.55$\,MB per client-round -- roughly $41\times$ less traffic over a
$50$-round schedule -- but need $37$ to $73$ sequential exchanges per federated
round, $1{,}850$ to $3{,}650$ over the schedule, each bound by network latency.
On one axis (Table~\ref{tab:crossover}, Appendix~\ref{exp:costmodel5}): at
$50$\,ms round-trip latency the interactive floor is $1.9$--$3.7$\,s per round
before any payload moves, while this protocol's $45.3$\,MB take $0.36$\,s at
$1$\,Gbps, $3.6$\,s at $100$\,Mbps and $36$\,s at $10$\,Mbps. The crossover sits
near $100$\,Mbps at wide-area latency: above it, the protocol is beneficial; below it or on a
low-latency metro link, the interactive baselines outweigh the benefit. The cross-institution regime
this work targets sits above the crossover; edge deployment does not. Precision is a second lever: traffic falls from $6{,}485$\,MB at
$12$ bits to $4{,}375$\,MB at $8$ bits with final loss statistically unchanged
(\S\ref{sec:exp-parity}). On compute, protocol accounting gives a $10\times$ to
$30\times$ per-rotation reduction against Solovay--Kitaev as $\epsilon$ tightens
from $10^{-2}$ to $10^{-6}$ (Appendix~\ref{exp:costmodel6}) -- a gate-count
estimate, not a wall-clock measurement. 
\vspace{-0.2cm}
\section{Discussion and limitations}
\label{dis-lim}
\vspace{-0.1cm}

Mechanism-level claims rest on proofs and configuration-insensitive measurements: the depth ledger holds on two independent backends, the aggregation bound spans two orders of magnitude of dispersion, and the entangler compilation is exact. Run-level claims provide feasibility evidence at a deliberately small scale, rigorously designed for falsifiability -- pairing is verified, and the most flattering interpretation of the empirical data was explicitly tested and rejected (\S\ref{sec:exp-parity}).

The protocol converts a latency problem into a throughput problem, offering a favourable trade above roughly $100$\,Mbps at wide-area latency (\S\ref{sec:comm}). Ciphertext compression represents the natural next step. This poses an equity concern alongside performance, given that environments requiring federated privacy are often the least well-connected. Separately, the encrypted-CNOT path costs $\approx1.5$\,h per gate in the current reference implementation (an implementation constraint that currently bounds delegated encrypted execution).

The composed guarantee is information-theoretic at the quantum-state level and computational under (Ring-)LWE on the classical layer, leaving the system bounded by the weaker link. The reported runs utilise a reference lattice profile chosen strictly for structural verification and carry no formal security guarantees (Appendix~\ref{app:security}). Decryption capability is shared across clients -- the standard shared-key assumption, with threshold key generation~\citep{mouchet2021multiparty} as a drop-in hardening -- whilst the honest-but-curious model precludes anomaly detection, leaving silent poisoning unaddressed. Generalising to larger client populations, further datasets, and alternative ansatz families remains future work, and hardware validation required access to a $156$-qubit processor (Appendix~\ref{app:broader-impact}).
\vspace{-0.2cm}
\section{Conclusion}

Encrypted training is usually framed as an architectural challenge, but for a growing class of models, it is a coordinate problem. In the spin chart, the composition law of $\SU(2)$ (and any $\Spin(n)$) is bilinear. Therefore, encrypted rotation updates cost one multiplicative level, federated averaging costs zero, and bootstrapping never enters the loop. The exportable principle is clear: when parameters reside on a group, the group law dictates encryption feasibility.
\vspace{-0.2cm}
\subsection*{Reproducibility statement}

Every number reported in this paper is produced by a script released upon acceptance or upon request. Appendix~\ref{app:experiments} records, per experiment,
the claim under test, the entry point, the exact command line, the seeds, the
cryptographic stack and precision, the hardware, and the job identifiers;
Appendix~\ref{app:security} records the lattice parameter profiles and states
which runs used which. Proposition~\ref{prop:degree2},
Proposition~\ref{prop:cubic} and Lemma~\ref{lem:entangler} are proved in
Appendices~\ref{sec:degree2}, \ref{app:geometry} and~\ref{app:correctness}
respectively. The two backend paths, the two aggregation compositions used by
the training harness and by the protocol, and the float64 decode setting used to
avoid a float32 cast floor in the exactness measurements are documented in
Appendix~\ref{app:experiments} so that any reported figure can be traced to the
code path that produced it. Appendix~\ref{app:llm} discloses the use of LLMs.
\vspace{-0.2cm}
\subsection*{Ethics statement}

This work strengthens federated confidentiality without introducing new surveillance or data extraction capabilities. Three structural limitations remain. Firstly, server blindness precludes update auditing, creating tension with model-inspection regimes and exposing the system to silent client poisoning. Secondly, physical validation requires institutional quantum and HPC access, risking the concentration of benefits among already well-resourced actors. Finally, the bandwidth profile disadvantages the exact low-connectivity participants for whom federated privacy guarantees are most valuable. These constraints and dual-use/compliance implications are detailed in Appendix~\ref{app:broader-impact}.

\bibliographystyle{iclr2027_conference}
\bibliography{iclr2027_conference}

\newpage
\appendix

\section*{Appendix overview}
\textbf{A}~quantum and federated-learning preliminaries;
\textbf{B}~proof of the degree-2 depth ledger and its $\Spin(n)$ generalisation;
\textbf{C}~aggregation correctness, the mask lemma, and the entangler
compilation proof;
\textbf{D}~aggregation geometry, with the proof of the cubic bound;
\textbf{E}~protocol details for Pauli-OTP and the quaternion track;
\textbf{F}~the detailed construction of encrypted QNN training;
\textbf{G}~experimental specifications, per experiment;
\textbf{H}~cost models;
\textbf{I}~scale tests;
\textbf{J}~implementation and security profiles;
\textbf{K}~extended discussion and limitations;
\textbf{L}~broader impact;
\textbf{M}~use of large language models.

\section{Quantum and Federated-Learning Preliminaries}
\label{app:preliminaries}

This appendix collects the minimum needed to follow \S\ref{sec:reduction}.
Readers fluent in federated learning can skip \S\ref{app:fl}; readers new to
quantum gates should read \S\ref{app:quantum}. Standard references are
\citet{nielsen2010quantum} for the quantum material and
\citet{kairouz2021advances} for federated learning.

\subsection{Federated learning}
\label{app:fl}

$K$ clients hold private datasets $\mathcal{D}_k$ of size $n_k$, with
$n=\sum_k n_k$. In each round every client initialises from the global
parameters, trains locally, and returns an update; the server forms
$\bar\vartheta = \sum_k \tfrac{n_k}{n}\vartheta^{(k)}$~\citep{mcmahan2017}. The
privacy failure that motivates encryption is that $\vartheta^{(k)}$ is
informative about $\mathcal{D}_k$: gradient-inversion attacks recover training
examples from transmitted updates alone~\citep{zhu2019deep,geiping2020inverting}.
Homomorphic aggregation removes the server from the trust boundary by making
the sum computable on ciphertexts~\citep{gentry2009fully,zhang2020batchcrypt}.

\subsection{Quantum preliminaries}
\label{app:quantum}

A qubit state is a unit vector in $\mathbb{C}^2$, written
$\ket{\psi}=\alpha\ket{0}+\beta\ket{1}$; $n$ qubits live in
$(\mathbb{C}^2)^{\otimes n}$, a $2^n$-dimensional space that does not factorise
into independent classical parts, which is what entanglement means
operationally. Gates are unitary matrices. The Pauli group is generated by
$X,Y,Z$; the Clifford group is the normaliser of the Pauli group, containing
$H$, $S$ and CNOT. Single-qubit rotations
$R_{\hat u}(\theta)=\exp(-i\theta\,\hat u\cdot\vec\sigma/2)$ are the trainable
weights of a variational circuit; they are not Clifford for generic $\theta$,
which is the entire source of difficulty in \S\ref{sec:intro}. Measurement in
the $Z$ basis returns a bit with probability given by the squared amplitudes,
so all reported quantities are estimated from finite shot counts and carry
sampling error $\propto 1/\sqrt{N}$.

A variational quantum circuit alternates rotation layers with a fixed
entangling pattern and is trained by gradient descent, with gradients from the
parameter-shift rule~\citep{mitarai2018quantum,schuld2019evaluating} or from
simulator backpropagation. Optimisation pathologies specific to this model class
-- barren plateaus~\citep{mcclean2018barren} -- are the reason a purely
quantum configuration underperforms the hybrid design
(Appendix~\ref{purequantum}).

\section{The Depth Ledger: Proofs}
\label{sec:degree2}

\begin{proof}[Proof of Proposition~\ref{prop:degree2}]
\emph{(i)} By Eq.~\ref{eq:quat-mul} each component of $q\otimes r$ is a sum of
terms $\pm q_j r_k$, each of total degree two, with coefficients
$M^{(i)}_{jk}\in\{-1,0,+1\}$ read off the Hamilton table; no other monomials
occur. \emph{(ii)} Substituting plaintext $r_k$ into Eq.~\ref{eq:quat-mul}
leaves a $\mathbb{R}$-linear form in the encrypted $q_j$ whose coefficients are
plaintext, so $\Enc(q')$ is obtained by plaintext--ciphertext multiplications
and additions, consuming one multiplicative level by assumption on
$\mathcal{E}$. \emph{(iii)} $\sum_k \tfrac{n_k}{n}\Enc\bigl(q^{(k)}\bigr)$ is an
$\mathbb{R}$-linear combination of ciphertexts with plaintext scalars, which
$\mathcal{E}$ supports at depth $0$; renormalisation is applied to the decrypted
plaintext and is therefore outside the homomorphic layer entirely. Composing
(ii) and (iii) gives the round bound.
\end{proof}

\begin{corollary}[Unbounded round count]
\label{cor:fedavg-depth}
A parameterisation of $\mathcal{E}$ supporting multiplicative depth $\geqslant1$
suffices for an unbounded number of federated rounds, each consisting of one
server-side plaintext rotation applied to every encrypted key followed by one
weighted aggregation, provided clients re-encrypt between rounds. No
bootstrapping is required in the inner loop.
\end{corollary}

\begin{proof}
Each round consumes at most one level by Proposition~\ref{prop:degree2}, and
re-encryption at the client resets the level budget at zero homomorphic cost
because clients already decrypt the aggregate in step 5 of the protocol.
\end{proof}

\begin{proof}[Proof of Corollary~\ref{cor:spin}]
Rotors $R\in\Spin(n)\subset\mathrm{Cl}^{[0]}(n)$ compose by the geometric
product, which is bilinear on the algebra: writing $R=\sum_A R_A e_A$ over the
even-graded basis $\{e_A\}$, the product $RR'$ has components
$\sum_{A,B} c^{C}_{AB} R_A R'_B$ with structure constants
$c^{C}_{AB}\in\{-1,0,+1\}$ fixed by the algebra~\citep{lounesto2001clifford}.
This is Eq.~\ref{eq:quat-mul} with a larger index set, and the argument of
Proposition~\ref{prop:degree2} applies unchanged. For $n=3$, the even subalgebra
is $\HH$ and the statement reduces to the quaternion case.
\end{proof}

\begin{remark}[Why not Euler angles]
The obstruction the reduction removes is worth naming precisely. In Euler
coordinates the composition law involves $\arctan$, $\arccos$ and products of
trigonometric functions of encrypted arguments. A levelled scheme reaches these
only through polynomial approximation, at a multiplicative depth that grows with
the required accuracy and typically forces bootstrapping inside the training
loop~\citep{cheon2017ckks}. In the discrete-alphabet route, the same composition is
exact but requires $\Theta(\log^{c}(1/\epsilon))$ gates per
rotation~\citep{Dawson2005SK}, in practice above $25{,}000$~\citep{Ma_2022},
each non-Clifford element demanding an active key repair. Degree two is
therefore not a marginal improvement over these routes; it is a different
complexity class for the same computation.
\end{remark}

\section{Correctness of Aggregation and Entangler Compilation}
\label{app:correctness}

\subsection{What the server computes, and why it is correct}

Let client $k$ hold trained angles $\theta^{(k)}_g$ for gate $g$, encoded as
$q^{(k)}_g=q(\theta^{(k)}_g)$ by Eq.~\ref{eq:quat-iso-app} and encrypted
component-wise under keys shared by the clients. The server computes
\begin{equation}
\Enc\bigl(q^{\text{glob}}_g\bigr) \;=\; \sum_{k=1}^{K} n_k \cdot \Enc\left(q^{(k)}_g\right),
\label{eq:agg}
\end{equation}
using ciphertext addition and plaintext-integer weighting only. Decryption
correctness of $\mathcal{E}$ gives
$\mathrm{Dec}\biggl(\Enc\bigl(q^{\text{glob}}_g\bigr)\biggr) = \sum_k n_k q^{(k)}_g + \varepsilon$
with $\|\varepsilon\|$ bounded by the scheme's encoding error, measured in
\S\ref{sec:exp-exact}; clients divide by $n$, renormalise and recover
$\hat\theta_g = 2\,\mathrm{atan2}(-t_1,t_0)$. No ciphertext--ciphertext
multiplication occurs, so no relinearisation or rescaling is needed and the
level budget is untouched. The server holds an evaluation-only context: the
aggregation object exposes addition, scaling and serialisation, and decryption
raises an error.

\subsection{Proof of the mask lemma}

\begin{proof}[Proof of Lemma~\ref{lem:mask}]
\emph{(i)} Quaternion multiplication is $\mathbb{R}$-bilinear, so for fixed $u$
the map $L_u$ is $\mathbb{R}$-linear and commutes with any finite
$\mathbb{R}$-linear combination:
$\sum_k w_k (u\otimes q_k) = u\otimes(\sum_k w_k q_k)$. Hence a mask shared by
all clients survives averaging intact and is removed by one plaintext
multiplication by $u^{-1}=\bar u$ after decryption.
\emph{(ii)} For the delegation channel, let $\rho_i,\rho_j$ be two clients'
states. Conjugation by a shared unitary is an isometry of the Hilbert--Schmidt
inner product,
$\mathrm{Tr}[(U\rho_iU^\dagger)(U\rho_jU^\dagger)]=\mathrm{Tr}(\rho_i\rho_j)$,
so a common mask hides each state individually but preserves all pairwise
overlaps, which is a non-trivial function of the clients' private parameters.
With independent uniform Pauli keys, the twirl
$\tfrac{1}{4}\sum_{a,b}(X^aZ^b\otimes I)(\rho_i\otimes\rho_j)(X^aZ^b\otimes
I)^\dagger$ applied to each factor sends the pair to
$\tfrac{I}{d}\otimes\tfrac{I}{d}$, so no function of the pair survives.
\end{proof}

The practical reading: masks on the \emph{aggregation} channel may be common
(and in the reported protocol are simply absent, since the aggregated object is
the unmasked update quaternion under shared keys), whereas masks on the
\emph{delegation} channel must be per-client and fresh.

\subsection{Proof of the entangler compilation lemma}

\begin{proof}[Proof of Lemma~\ref{lem:entangler}]
Exact decompositions of controlled rotations into at most two CNOTs plus
single-qubit gates are given by~\citet{barenco1995elementary}, and of a general
two-qubit unitary into at most three CNOTs by the Cartan decomposition
argument of~\citet{vatan2004optimal}; both are exact identities, not
approximations. Applying either places the continuous parameter on a
single-qubit rotation, which by Proposition~\ref{prop:degree2}(ii) costs one
multiplicative level. The CNOTs are Clifford, so conjugation maps the Pauli
frame to another Pauli frame; for $X^{a_1}Z^{b_1}\otimes X^{a_2}Z^{b_2}$ the
induced key map is
$(a_1,b_1,a_2,b_2)\mapsto(a_1,\,b_1\oplus b_2,\,a_1\oplus a_2,\,b_2)$, a fixed
XOR relabelling of pad bits carrying no multiplicative depth and no dependence
on $\theta$. Finally $R_Z(\theta)X^{a}=X^{a}R_Z((-1)^{a}\theta)$, so passing the
rotation through the pad conditions only the sign of the angle on the key.
Hence the number of encrypted operations per entangler is a constant
independent of $\theta$ and of the ansatz, and the rotation-track depth is
unchanged.
\end{proof}

\begin{remark}[Consequence for expressivity]
Because the decompositions are exact, no expressivity is lost: any
parameterised entangler an unrestricted ansatz could use is available here at a
fixed gate-count factor. The depth class -- and therefore the absence of
bootstrapping -- depends only on the number of \emph{sequential} rotations. In
the federated loop it cannot grow at all, because clients train locally and
encrypt only the resulting angles.
\end{remark}

\section{Aggregation Geometry}
\label{app:geometry}

\begin{proposition}[Aggregation geometry]
\label{prop:cubic}
Fix an axis $\hat u$ and weights $w_k\geqslant0$, $\sum_k w_k=1$, and let
$\bar\theta=\sum_k w_k\theta_k$ with deviations
$\delta_k=\theta_k-\bar\theta$. Let $\hat\theta$ be the angle recovered from the
renormalised weighted quaternion mean. Then
\begin{equation*}
\hat\theta - \bar\theta \;=\; -\frac{1}{24}\sum_k w_k \delta_k^3 \;+\;
O(\delta^5),
\end{equation*}
provided all representatives lie in a common hemisphere. The bound fails, by
exactly $\pi$, for configurations straddling the branch cut, which the wrapping
step above excludes by construction.
\end{proposition}

\begin{proof}[Proof of Proposition~\ref{prop:cubic}]
Fix the axis and write $\varphi_k=\theta_k/2$, $\bar\varphi=\bar\theta/2$ and
$\delta_k=\theta_k-\bar\theta$, so $\varphi_k=\bar\varphi+\delta_k/2$ and
$\sum_k w_k\delta_k=0$. The weighted quaternion mean has components
\begin{align}
t_0 &= \sum_k w_k\cos\varphi_k = \cos\bar\varphi\, m_c - \sin\bar\varphi\, m_s,
&
-t_1 &= \sum_k w_k\sin\varphi_k = \sin\bar\varphi\, m_c + \cos\bar\varphi\, m_s,
\end{align}
with $m_c=\sum_k w_k\cos(\delta_k/2)$ and $m_s=\sum_k w_k\sin(\delta_k/2)$.
Therefore
$\mathrm{atan2}(-t_1,t_0) = \bar\varphi + \psi$ with
$\tan\psi = m_s/m_c$. Expanding,
$m_s = \tfrac12\sum_k w_k\delta_k - \tfrac{1}{48}\sum_k w_k\delta_k^3 +
O(\delta^5) = -\tfrac{1}{48}\sum_k w_k\delta_k^3 + O(\delta^5)$
using $\sum_k w_k\delta_k=0$, and
$m_c = 1 - \tfrac18\sum_k w_k\delta_k^2 + O(\delta^4)$. Hence
$\psi = -\tfrac{1}{48}\sum_k w_k\delta_k^3 + O(\delta^5)$ and
$\hat\theta = 2(\bar\varphi+\psi) = \bar\theta - \tfrac{1}{24}\sum_k
w_k\delta_k^3 + O(\delta^5)$, which is Eq.~\ref{eq:cubic}. Renormalisation does
not affect the recovered angle because $\mathrm{atan2}$ is scale-invariant. The
expansion requires the representatives to share a hemisphere; if two groups
straddle the branch cut, $m_c\to0$ and the recovered angle jumps by $\pi$.
\end{proof}

The factor $\tfrac{1}{24}$, rather than the $\tfrac16$ of a naive full-angle
expansion, is the visible trace of the half-angle encoding: the protocol
averages quaternion components at $\theta/2$.

\paragraph{Measured validation.}
Table~\ref{tab:e3} reports the dispersion sweep of \S\ref{sec:exp-geometry}: five
clients, $36$ angles each, five repetitions per dispersion, through the secure
stack at $12$-bit precision. ``Exact bias'' is the chordal-versus-arithmetic
difference computed on plaintext angles, isolating geometry from encryption
noise; ``prediction'' is Eq.~\ref{eq:cubic}; ``floor'' is the encryption noise
of the same runs.

\begin{table}[H]
\caption{Aggregation-geometry sweep. The cubic bound tracks the realised bias
across two orders of magnitude of client dispersion, and the bias is below the
encryption noise floor throughout.}
\label{tab:e3}
\centering
\small
\begin{tabular}{@{}lcccc@{}}
\toprule
$\sigma$ (rad) & Exact bias (rad) & Prediction (rad) & Correlation & Encryption floor (rad) \\
\midrule
$0.05$ & $2.4956\times10^{-6}$ & $2.4949\times10^{-6}$ & $1.0000$ & $1.1705\times10^{-4}$ \\
$0.20$ & $1.4067\times10^{-4}$ & $1.4014\times10^{-4}$ & $1.0000$ & $2.0381\times10^{-4}$ \\
$0.50$ & $2.1721\times10^{-3}$ & $2.1126\times10^{-3}$ & $0.9999$ & $2.1764\times10^{-3}$ \\
$1.00$ & $1.7974\times10^{-2}$ & $1.6236\times10^{-2}$ & $0.9991$ & $1.7996\times10^{-2}$ \\
\bottomrule
\end{tabular}
\end{table}

\paragraph{The branch-cut case, stated precisely.}
With two client groups placed at $\bar\theta\pm(\pi-\epsilon)$, $\epsilon=0.05$,
the naive arithmetic mean of angles errs by $3.1416$\,rad against the true
circular mean, and the exact chordal mean by $0.0000$\,rad. Two compositions
appear in our codebase and it is worth separating them explicitly: the
\emph{protocol} path aggregates encrypted quaternions and recovers the angle by
$\mathrm{atan2}$, hence computes the chordal mean; the \emph{scale-test harness}
composes a per-client encrypted round trip with masked secure aggregation of the
recovered angles, hence computes the Euclidean mean and inherits the $\pm\pi$
failure in this constructed case ($3.1414$\,rad). Neither is exercised
adversarially in any trained run, because the decode step re-wraps into
$(-\pi,\pi]$ every round; the principal-branch and hemisphere-alignment steps of
\S\ref{sec:geometry} close the case structurally, in plaintext, at zero depth
and zero communication. Since the two means differ by
$2.50\times10^{-6}$\,rad at the operating dispersion -- $47\times$ below the
encryption noise floor -- no reported number depends on which is used.

\section{Quantum Homomorphic Encryption Details}
\label{app:qotp}

\subsection{Pauli One-Time Pad}
\label{potp}

Pauli-OTP~\citep{Broadbent_2015} encrypts $\ket{\psi}$ as
$\ket{\psi'} = X^aZ^b\ket{\psi}$ with $a,b\in\{0,1\}^n$ and
$X^aZ^b=(X^{a_1}Z^{b_1})\otimes\cdots\otimes(X^{a_n}Z^{b_n})$. Keys are sampled
fresh per use and shared with the server only in encrypted form. Averaged over
the key the ciphertext is maximally mixed, so the scheme is
information-theoretically secure.

\subsubsection{Clifford gates}
\label{potp:clifford}

For Clifford $C$ and Pauli $P$, $CPC^\dagger\in\mathcal{P}$, so
\begin{align}
C\ket{\psi'} = C X^aZ^b\ket{\psi} = (CX^aZ^bC^\dagger)(C\ket{\psi})
            = X^{a'}Z^{b'}(C\ket{\psi}),
\end{align}
and only the classical keys change. The server updates the encrypted keys
purely classically; no client interaction is required.

\subsubsection{Non-Clifford gates}
\label{potp:non-clifford}

For the $T$ gate,
\begin{align}
T\ket{\psi'} = T(X^aZ^b\ket{\psi}) = X^aZ^{a\oplus b}P^{a}(T\ket{\psi}),
\label{eq:phaseerr}
\end{align}
so the residual error $P^{a}$ is conditioned on a secret bit. The accepted
repairs are deferred correction, $O(N_T^2)$ in the number of $T$
gates~\citep{Broadbent_2015}, or auxiliary magic states with resource cost
exponential in $T$-depth~\citep{Broadbent2009Universal}. For arbitrary
continuous rotations the Solovay--Kitaev route requires above $25{,}000$ gates
per rotation~\citep{Dawson2005SK,Ma_2022}. This is the obstruction that
\S\ref{sec:degree2-main} removes by changing coordinates rather than by paying
either price.

\subsubsection{Encrypted CNOT}
\label{ecnot}

\citet{mahadev2023} constructs an encrypted-CNOT primitive from trapdoor
claw-free functions under LWE, allowing a fully classical client: the server
prepares a superposition over function inputs, entangles it with the target and
measures the output, so the conditional flip is controlled by an encrypted bit
without decryption.

\subsection{Quaternion track}
\label{sec:QOTP}

A quaternion $q=t_1+t_2\mathbf{i}+t_3\mathbf{j}+t_4\mathbf{k}$ with
$\|t\|_2=1$ maps to a single-qubit unitary by
\begin{equation}
U(t) = t_1 I_2 + t_2(iX) + t_3(iZ) + t_4(iY),
\end{equation}
so the encrypted object is a constrained real $4$-vector, which is why an
approximate-arithmetic scheme over reals is the natural backend. On the
delegation channel the key is a unit quaternion and the masked state
$U(t)\ket{\psi}$ is a uniformly random point on the Bloch sphere. When the
server applies a rotation $V$ the key updates as
$U(t_{\text{new}}) = U(t)V^{-1}$, a degree-2 homomorphic update
(Proposition~\ref{prop:degree2}).

Two-qubit gates couple the keys,
\begin{equation}
\mathrm{CNOT}\bigl(U(t_1)\otimes U(t_2)\bigr)
= \bigl(U(t_1')\otimes U(t_2')\bigr)\mathrm{CNOT}(\ket{\psi_1\psi_2}),
\label{qotp:cnot:equation}
\end{equation}
preventing independent per-qubit updates. Following~\citet{Ma_2022} the state is
converted to the Pauli frame immediately before an entangling gate, where the
update is the XOR relabelling of Lemma~\ref{lem:entangler}, and converted back
afterwards. The conversion decomposes the encrypted angle into Euler form and
approximates each angle bit-wise, applying conditional rotations under the
encrypted-CNOT primitive; its bit-serial character is its cost, and it is
amortised across the rotation layers it enables. The frame conversions are a
constant number per entangler and are levelled, requiring no periodic carry
bootstrapping (Appendix~\ref{app:security}).

\section{Encrypted QNN Training: Detailed Construction}
\label{app:encrypted-qnn-training}

\subsection{Client side}
\label{app:framework-client}

A trainable rotation
$U(\theta,\hat u)=\cos(\theta/2)I-i\sin(\theta/2)(u_xX+u_yY+u_zZ)$ is encoded by
Eq.~\ref{eq:quat-iso-app}. In the federated path the client encrypts the four
components of the \emph{unmasked} update quaternion under the clients' shared
homomorphic keys, after principal-branch wrapping and hemisphere alignment
(\S\ref{sec:geometry}). Quaternion one-time-pad keys are a separate object used
only on the delegation channel, per client and per use, in accordance with
Lemma~\ref{lem:mask}.

\subsection{Server side}
\label{app:framework-server}

The server applies plaintext rotations by right multiplication,
$\Enc(q_{\text{new}}) = \Enc(q)\cdot V^{-1}$ at one multiplicative level, and
forms the weighted sum of Eq.~\ref{eq:agg} at depth $0$. It never decrypts and
never multiplies ciphertext by ciphertext.

\subsection{Entangling gates}
\label{app:framework-cnot}

Entangling gates are handled by the short representation switch of
\S\ref{sec:QOTP}, yielding the division of labour that makes the construction
work: continuous parameters ride the quaternion track, where composition is
degree $2$; Clifford structure rides the Pauli track, where updates are XOR
relabellings; and the two meet only at frame conversions of constant count.

\section{Experimental Specifications}
\label{app:experiments}

For each experiment we record: the claim under test, the procedure and entry
point, the explicit assumptions, and the observed result. Unless stated
otherwise, runs use the secure fixed-point multi-key stack at $12$-bit
quaternion precision, a six-qubit six-layer ansatz with a CNOT ring after each
rotation layer ($36$ trainable angles), and exact statevector simulation with
parameter-shift gradients. Cluster jobs were executed on NVIDIA T4, A40 and
A100 nodes; job identifiers are given per experiment so that logs can be matched
to results.

\paragraph{Two code paths, named once.}
Two aggregation compositions exist in the release and are used for different
purposes. The \emph{protocol path} aggregates encrypted quaternions and recovers
angles by $\mathrm{atan2}$ (chordal mean); the \emph{scale-test harness}
composes a per-client encrypted round trip with masked secure aggregation of
recovered angles using pairwise-cancelling masks and Shamir
sharing~\citep{shamir1979share,diffie1976new} for dropout tolerance (Euclidean
mean). The server decrypts in neither. Depth-0 aggregation claims refer to the
protocol path and are measured in Appendix~\ref{exp:e2}; federated learning
curves come from the harness; Appendix~\ref{app:geometry} quantifies the
difference between the two as $2.50\times10^{-6}$\,rad at the operating
dispersion.

\paragraph{Decode precision.}
Exactness measurements decode in float64. This matters: a float32 decode path
pins the residual near $10^{-7}$ relative precision and therefore cannot
represent the true CKKS aggregation error, which is four orders of magnitude
smaller. All values in Table~\ref{tab:exact} are float64 decodes.

\begin{figure}[H]
    \centering
    \begin{subfigure}[t]{0.325\textwidth}
        \centering
        \includegraphics[width=\linewidth]{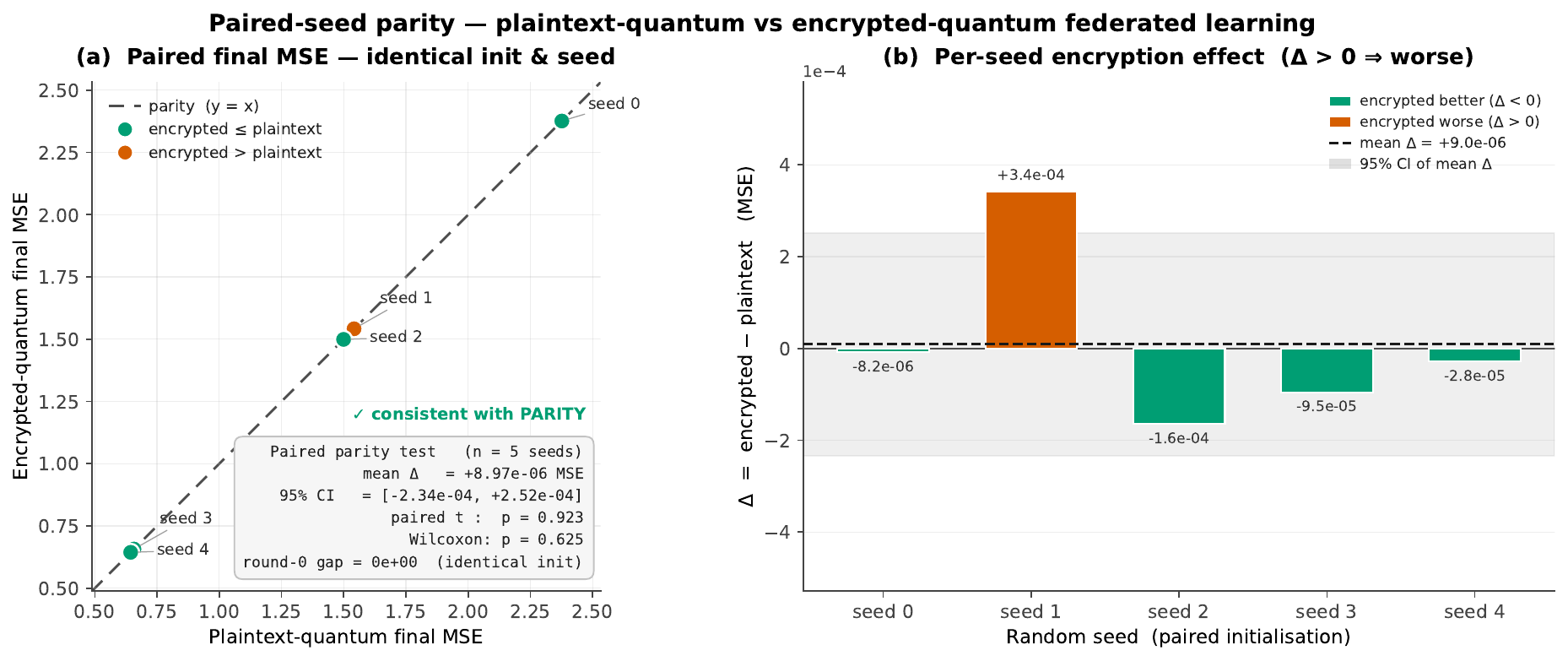}
        \caption{Paired parity: encrypted vs.\ plaintext, five seeds.}
        \label{fig:e1}
    \end{subfigure}\hfill
    \begin{subfigure}[t]{0.325\textwidth}
        \centering
        \includegraphics[width=\linewidth]{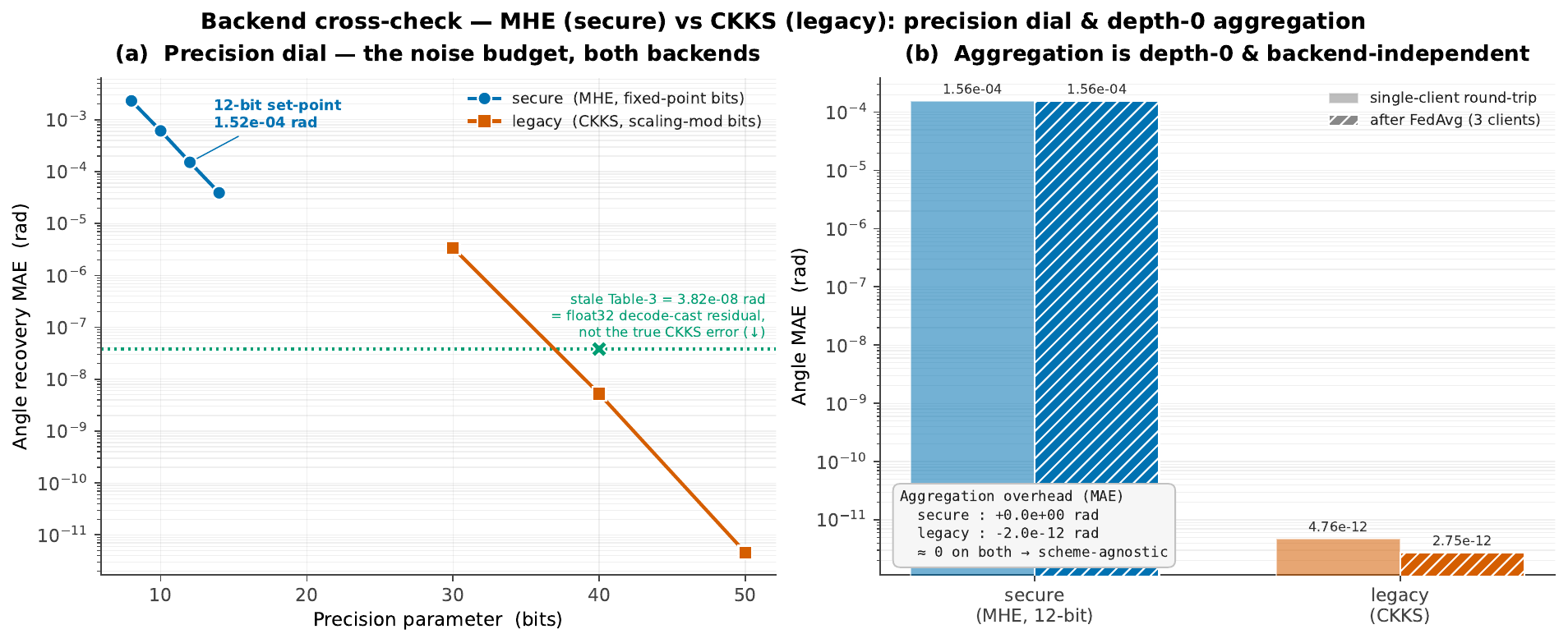}
        \caption{Precision dial and depth-0 aggregation, two backends.}
        \label{fig:e2}
    \end{subfigure}\hfill
    \begin{subfigure}[t]{0.325\textwidth}
        \centering
        \includegraphics[width=\linewidth]{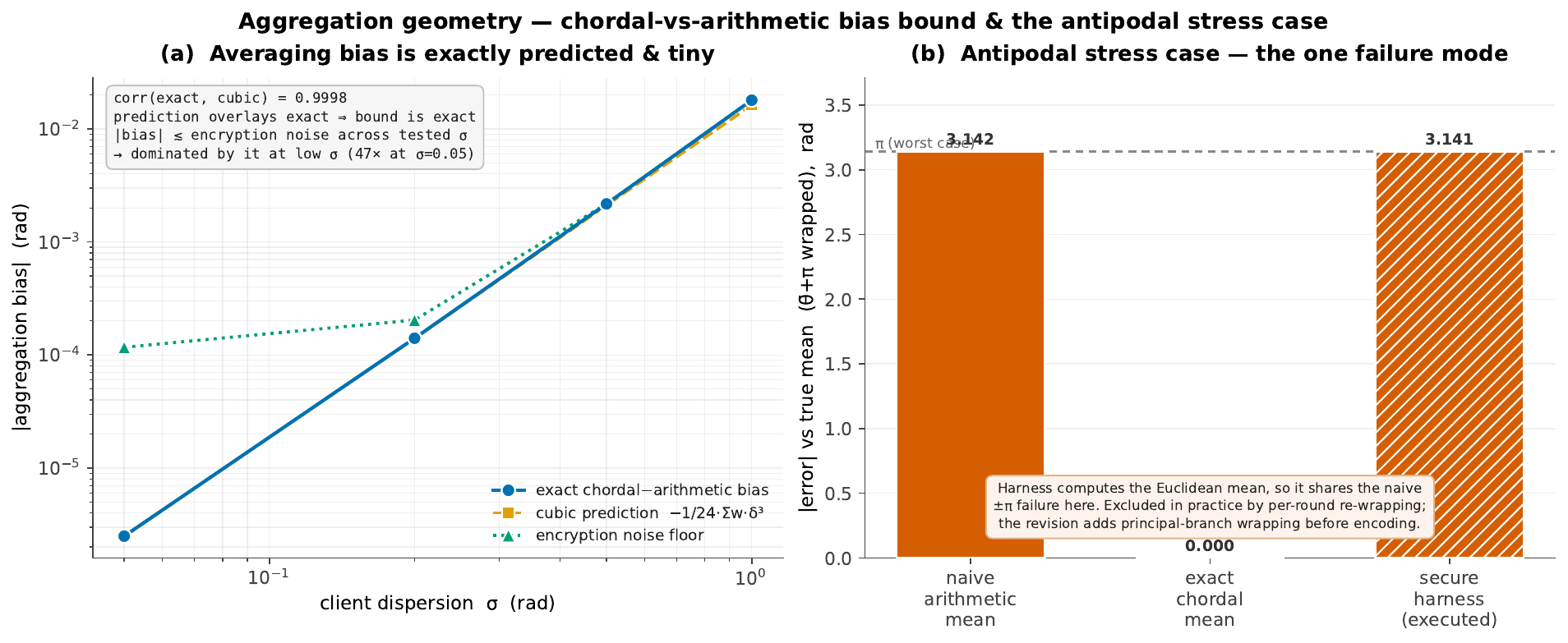}
        \caption{Aggregation bias against the cubic bound of Eq.~\ref{eq:cubic}.}
        \label{fig:e3}
    \end{subfigure}
    \caption{The three load-bearing measurements. (a)~With initialisation and
    batch order held identical, encrypted and plaintext arms are
    indistinguishable. (b)~Weighted FedAvg adds no error beyond per-client
    encoding on either backend. (c)~The realised aggregation bias is predicted by
    Eq.~\ref{eq:cubic} and sits below the encryption noise floor throughout.}
    \label{fig:results}
\end{figure}

\subsection{E1: Paired-seed parity}
\label{exp:e1}

\paragraph{Claim.} With initialisation and data order held identical, encrypted
and plaintext federated training reach statistically indistinguishable final
loss.
\paragraph{Procedure.} For each seed in $\{0,1,2,3,4\}$, reconfigure
reproducibility and rebuild the data loaders \emph{immediately before each arm},
then run the plaintext and encrypted arms in turn; five federated rounds, three
clients, secure stack, $12$-bit. Report per-seed final validation MSE, a paired
$t$-test, a Wilcoxon signed-rank test and a $95\%$ confidence interval.
\paragraph{Assumptions and controls.} Pairing is verified rather than assumed:
arms that begin identically must have equal round-0 loss. The maximum round-0
gap across seeds is $0.00\times10^{0}$. Exact noiseless simulation is used
deliberately, since injecting device noise would confound the quantity being
isolated.
\paragraph{Result.} Plaintext $1.3436\pm0.7228$, encrypted $1.3436\pm0.7229$;
per-seed differences $\{-8.2\times10^{-6}, +3.4\times10^{-4}, -1.7\times10^{-4},
-9.5\times10^{-5}, -2.8\times10^{-5}\}$; mean $+9.0\times10^{-6}$, $95\%$ CI
$[-2.3\times10^{-4}, +2.5\times10^{-4}]$; $t=0.103$, $p=0.923$; Wilcoxon
$W=5.0$, $p=0.625$. Job \texttt{265819}. The between-seed spread of final MSE is
$\approx0.72$ (individual finals range $0.64$ to $2.38$), which is why single-run
comparisons in this regime are uninformative about the encryption effect and why
the paired design is necessary rather than merely tidy.

\begin{table}[H]
\caption{Utility under encryption. Five paired seeds, five federated rounds,
three clients; final validation MSE (mean $\pm$ s.d.). Pairing verified: maximum
round-0 gap $=0.00\times10^{0}$ across all seeds. Rows 1--2 isolate the
encryption effect; rows 2--3 the noise-budget effect.}
\centering
\footnotesize
\begin{tabular}{@{}lccc@{}}
\toprule
\textbf{Arm} & \textbf{Final MSE} & \textbf{Paired $\Delta$ (95\% CI)} & \textbf{Paired $t$ / Wilcoxon} \\
\midrule
Plaintext                   & $1.3436\pm0.7228$ & -- & -- \\
Encrypted, 12-bit           & $1.3436\pm0.7229$ & $+9.0{\times}10^{-6}\;[-2.3,+2.5]{\times}10^{-4}$ & $p{=}0.923$ / $p{=}0.625$ \\
Encrypted, 8-bit ($16\times$ noise) & $1.3442\pm0.7230$ & $+6.0{\times}10^{-4}\;[-2.6,+3.8]{\times}10^{-3}$ & $p{=}0.636$ / $p{=}1.0$ \\
\bottomrule
\end{tabular}
\end{table}

\subsection{E2: Backend cross-check and the noise budget}
\label{exp:e2}

\paragraph{Claim.} Weighted FedAvg over encrypted quaternions adds no error
beyond per-client encoding, on cryptographically distinct backends; and the
encryption noise floor is a tunable dial.
\paragraph{Procedure.} Encrypt three clients' $36$-angle vectors, aggregate
homomorphically, decrypt, and measure single-client versus post-aggregation
angle MAE; the difference is the aggregation overhead. Repeat on the secure
multi-key stack and on OpenFHE CKKS (ring $2^{16}$, depth $40$, $50$-bit
scaling, \texttt{FLEXIBLEAUTO}, \texttt{HYBRID} key switching). Sweep precision
($\texttt{qotp\_bits}\in\{8,10,12,14\}$; CKKS scaling modulus
$\in\{30,40,50\}$), $50$ repetitions per point, seed $42$, float64 decode.
\paragraph{Assumptions.} Both stacks are genuine cryptographic engines rather
than mocks; the CKKS arm performs real context generation and is memory-heavy.
\paragraph{Result.} Table~\ref{tab:exact}. Aggregation overhead
$0.0$ (secure) and $-2.0089\times10^{-12}$\,rad (CKKS); the negative sign is
independent rounding errors averaging down. Job \texttt{265860}.

\begin{table}[H]
\caption{Depth-0 aggregation is exact on both backends. The difference between
the two middle columns is the overhead Proposition~\ref{prop:degree2}(iii)
predicts to be zero. Float64 decode, three clients, $36$ angles, radians;
precision sweep in Table~\ref{tab:precision}.}
\centering
\small
\begin{tabular}{@{}lccc@{}}
\toprule
\textbf{Backend} & \textbf{Single client} & \textbf{Post-FedAvg} & \textbf{Aggregation overhead} \\
\midrule
Fixed-point multi-key (12-bit) & $1.5632\times10^{-4}$ & $1.5632\times10^{-4}$ & $\mathbf{0.0}$ \\
OpenFHE CKKS (50-bit modulus)  & $4.7573\times10^{-12}$ & $2.7484\times10^{-12}$ & $\mathbf{-2.0\times10^{-12}}$ \\
\bottomrule
\end{tabular}
\end{table}

\begin{table}[H]
\caption{Encryption noise budget as a function of the precision knob, float64
decode, $50$ repetitions per point. Angle MAE in radians.}
\label{tab:precision}
\centering
\small
\begin{tabular}{@{}lcccc@{}}
\toprule
Multi-key, precision bits & 8 & 10 & 12 & 14 \\
Angle MAE (rad) & $2.31\times10^{-3}$ & $6.11\times10^{-4}$ & $1.52\times10^{-4}$ & $3.91\times10^{-5}$ \\
\midrule
CKKS, scaling modulus & 30 & 40 & 50 & -- \\
Angle MAE (rad) & $3.39\times10^{-6}$ & $5.27\times10^{-9}$ & $4.58\times10^{-12}$ & -- \\
\bottomrule
\end{tabular}
\end{table}

\subsection{E3: Aggregation-geometry stress test}
\label{exp:e3}

\paragraph{Claim.} The chordal-versus-Euclidean aggregation bias is predicted by
Eq.~\ref{eq:cubic}, is negligible in the operating regime, and fails only in a
constructed branch-cut configuration that the protocol excludes.
\paragraph{Procedure.} Five clients, $36$ angles, dispersions
$\sigma\in\{0.05,0.2,0.5,1.0\}$\,rad, five repetitions each, through the secure
round-trip and aggregation primitives; plus a near-antipodal case with two
groups at $\bar\theta\pm(\pi-0.05)$. Seed $42$.
\paragraph{Result.} Table~\ref{tab:e3}; antipodal case $3.1416$\,rad (arithmetic
mean), $0.0000$\,rad (chordal mean), $3.1414$\,rad (harness). The run is
deterministic: an independent re-execution on different hardware reproduces
every reported digit.

\subsection{E4: Parameterised entangler through the encrypted pipeline}
\label{exp:e4}

\paragraph{Claim.} $R_{ZZ}(\theta)$ compiles exactly onto the supported gate set
and traverses the encrypted pipeline within the shot-noise floor, at constant
gate overhead.
\paragraph{Procedure.} Execute
$R_{ZZ}(\theta)=\mathrm{CNOT}\cdot(I\otimes R_Z(\theta))\cdot\mathrm{CNOT}$
end to end through the encrypted stack -- encrypted state preparation,
encrypted CNOT, encrypted rotation, encrypted CNOT, decrypt and run -- for $20$
angles spanning $[0,\pi]$ at $4096$ shots, recovering $\theta$ from the measured
distribution. Seed $42$, $12$-bit precision.
\paragraph{Result.} Mean recovered-angle error $8.33\times10^{-3}$\,rad, maximum
$2.06\times10^{-2}$\,rad, against a sampling standard deviation of
$1.56\times10^{-2}$\,rad and an expected mean absolute error of
$1.25\times10^{-2}$\,rad under pure sampling noise; worst trial at $1.3$
standard deviations. Exactly two CNOTs per trial, independent of $\theta$;
disentangling sanity check (second qubit returns to $\ket{0}$) passes with
probability $1.000$ in all $20$ trials. Timing: encryption $6.7$\,ms,
decryption $63.6$\,ms, encrypted evaluation $\approx3.06$\,h per trial, i.e.\
$\approx1.5$\,h per encrypted CNOT.

\begin{figure}[H]
    \centering
    \includegraphics[width=0.8\linewidth]{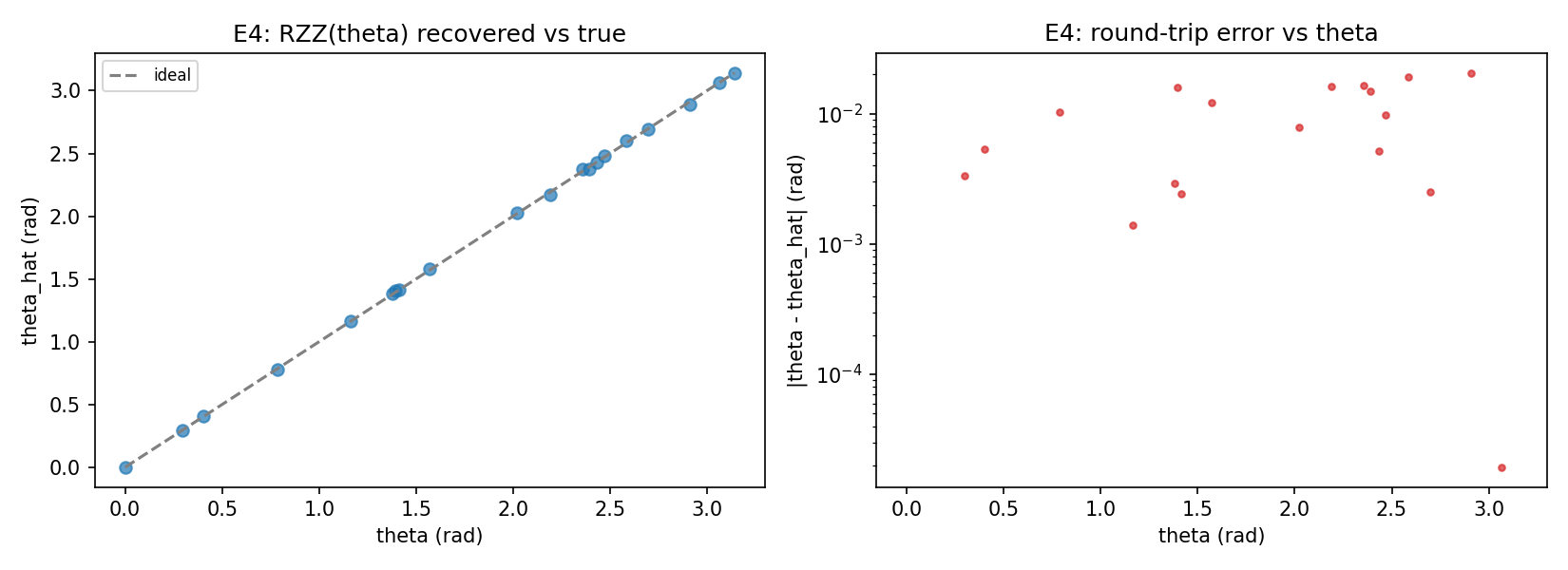}
    \caption{E4: recovered versus true $\theta$ for $R_{ZZ}(\theta)$ through the
    encrypted pipeline, with the shot-noise band. The encrypted entangler adds
    no angle error resolvable at $4096$ shots.}
    \label{fig:e4}
\end{figure}

\subsection{E5: Noise-budget falsification}
\label{exp:e5}

\paragraph{Claim.} Encryption noise does not act as a beneficial regulariser.
\paragraph{Procedure.} Re-run the encrypted arm at $8$-bit precision -- a
$16\times$ larger quantisation step -- using E1's per-seed methodology so that
each $8$-bit run starts from identical weights and data order to the
corresponding $12$-bit run. The $12$-bit values are read from E1's result file
rather than re-measured, so the comparison is genuinely paired.
\paragraph{Result.} $8$-bit $1.3442\pm0.7230$ against $12$-bit
$1.3436\pm0.7229$; per-seed differences $-0.0003$, $+0.0048$, $-0.0016$,
$-0.0012$, $+0.0013$; mean $+0.0006$, $95\%$ CI $[-0.0026,+0.0038]$;
$t=0.511$, $p=0.636$; Wilcoxon $W=7.0$, $p=1.0$. The
prediction implied by the regulariser reading -- more noise, more benefit --
does not hold. Job \texttt{266019}.

\begin{figure}[H]
    \centering
    \includegraphics[width=0.8\linewidth]{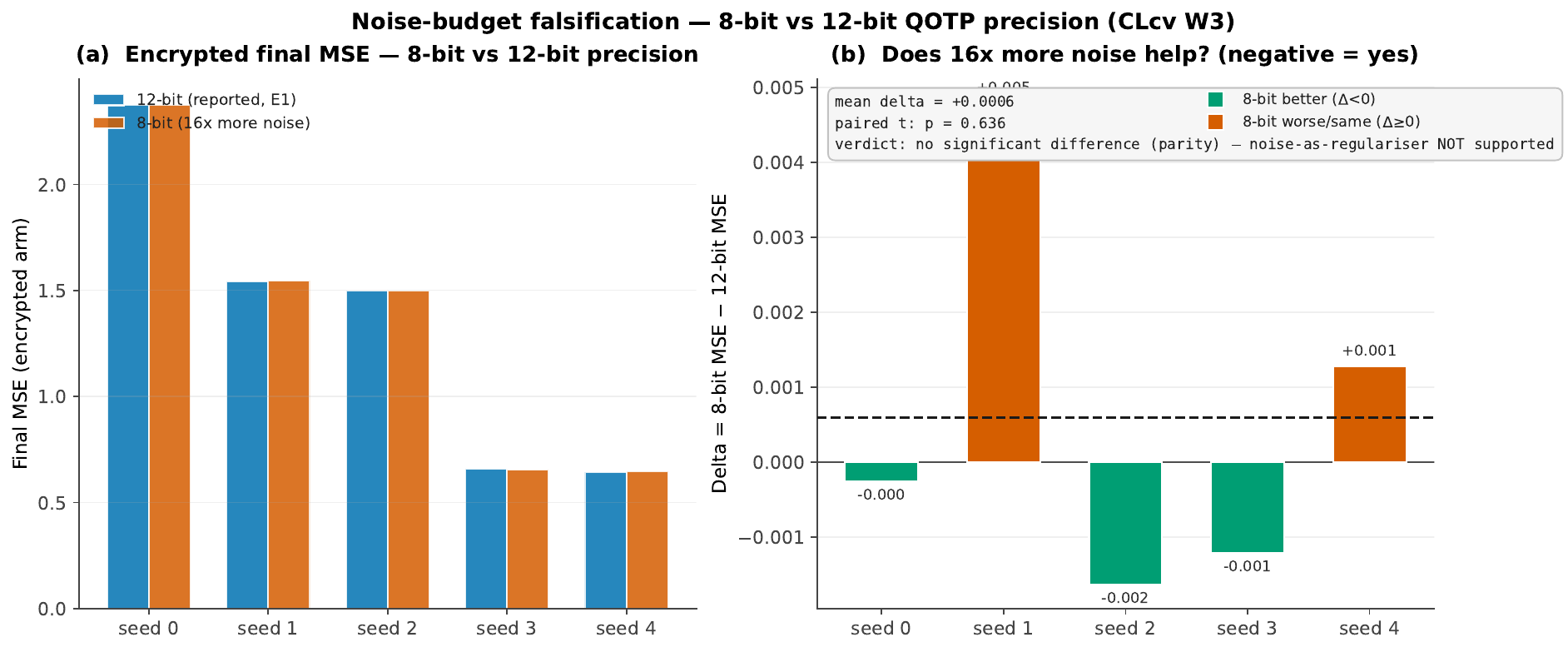}
    \caption{E5: $8$-bit versus $12$-bit encrypted arms, paired by seed. A
    $16\times$ increase in quantisation noise produces no detectable change in
    final loss.}
    \label{fig:e5}
\end{figure}

\subsection{E6: Cross-dataset scale study}
\label{exp:e6}

\paragraph{Claim.} Client scaling, horizon and precision--bandwidth behaviour
are not artefacts of one dataset or of the originally tested client range.
\paragraph{Procedure.} Three sub-sweeps -- client scaling
$K\in\{3,5,10,20\}$ at $10$ rounds; horizon $\in\{20,35,50\}$ rounds; precision
$\in\{8,10,12,14\}$ bits -- run identically on California
Housing~\citep{pace1997sparse} and Wine Quality~\citep{cortez2009modeling}, five
seeds per point, secure stack, $12$-bit unless swept.
\paragraph{Result.} Table~\ref{tab:scale} and Table~\ref{tab:bandwidth}. Growth
with client count is mild and monotone ($5.6\%$ and $7.3\%$ for a $6.7\times$
increase in participants); longer horizons improve both the mean and the seed
spread; the precision--bandwidth relation is identical across datasets, as a
protocol-level property should be. Jobs \texttt{265935}, \texttt{266315}
(California, split across two submissions after an unrelated cluster filesystem
failure; configuration parity between the two is recorded in the result
metadata) and \texttt{265936} (Wine, single job).

\begin{table}[H]
\caption{Cross-dataset scale study, final validation MSE (mean $\pm$ s.d.\ over
five seeds), secure stack at $12$-bit precision. Client scaling at $10$
federated rounds; horizon sweep at three clients.}
\label{tab:scale}
\centering
\small
\begin{tabular}{@{}lcccc|ccc@{}}
\toprule
& \multicolumn{4}{c|}{\textbf{Clients $K$}} & \multicolumn{3}{c}{\textbf{Rounds}} \\
\textbf{Dataset} & $3$ & $5$ & $10$ & $20$ & $20$ & $35$ & $50$ \\
\midrule
California Housing & $1.3035$ & $1.3166$ & $1.3583$ & $1.3763$ & $1.2285$ & $1.1395$ & $1.0831$ \\
 & \scriptsize$\pm0.614$ & \scriptsize$\pm0.633$ & \scriptsize$\pm0.659$ & \scriptsize$\pm0.670$ & \scriptsize$\pm0.541$ & \scriptsize$\pm0.444$ & \scriptsize$\pm0.374$ \\
Wine Quality & $1.4214$ & $1.4641$ & $1.5034$ & $1.5246$ & $1.3148$ & $1.1963$ & $1.1275$ \\
 & \scriptsize$\pm0.242$ & \scriptsize$\pm0.249$ & \scriptsize$\pm0.254$ & \scriptsize$\pm0.255$ & \scriptsize$\pm0.205$ & \scriptsize$\pm0.141$ & \scriptsize$\pm0.095$ \\
\bottomrule
\end{tabular}
\end{table}

\begin{table}[H]
\caption{Precision, ciphertext size, reconstruction error and total protocol
traffic over a $50$-round, three-client schedule. Traffic and ciphertext size are
dataset-independent by construction and reproduce identically on both datasets.
Angle MAE here is measured on $256$ probe angles; the independent float64
measurement of Table~\ref{tab:exact} gives $1.52\times10^{-4}$ at $12$ bits,
consistent within the probe-set difference.}
\label{tab:bandwidth}
\centering
\small
\begin{tabular}{@{}lcccc@{}}
\toprule
Precision bits & 8 & 10 & 12 & 14 \\
\midrule
Ciphertext size & $100.0$\,KB & $125.0$\,KB & $150.0$\,KB & $175.0$\,KB \\
Angle MAE (rad) & $2.600\times10^{-3}$ & $6.585\times10^{-4}$ & $1.676\times10^{-4}$ & $4.051\times10^{-5}$ \\
Total traffic & $4{,}375$\,MB & $5{,}430$\,MB & $6{,}485$\,MB & $7{,}539$\,MB \\
\bottomrule
\end{tabular}
\end{table}

\begin{figure}[H]
    \centering
    \begin{subfigure}[t]{0.48\textwidth}
        \centering
        \includegraphics[width=\linewidth]{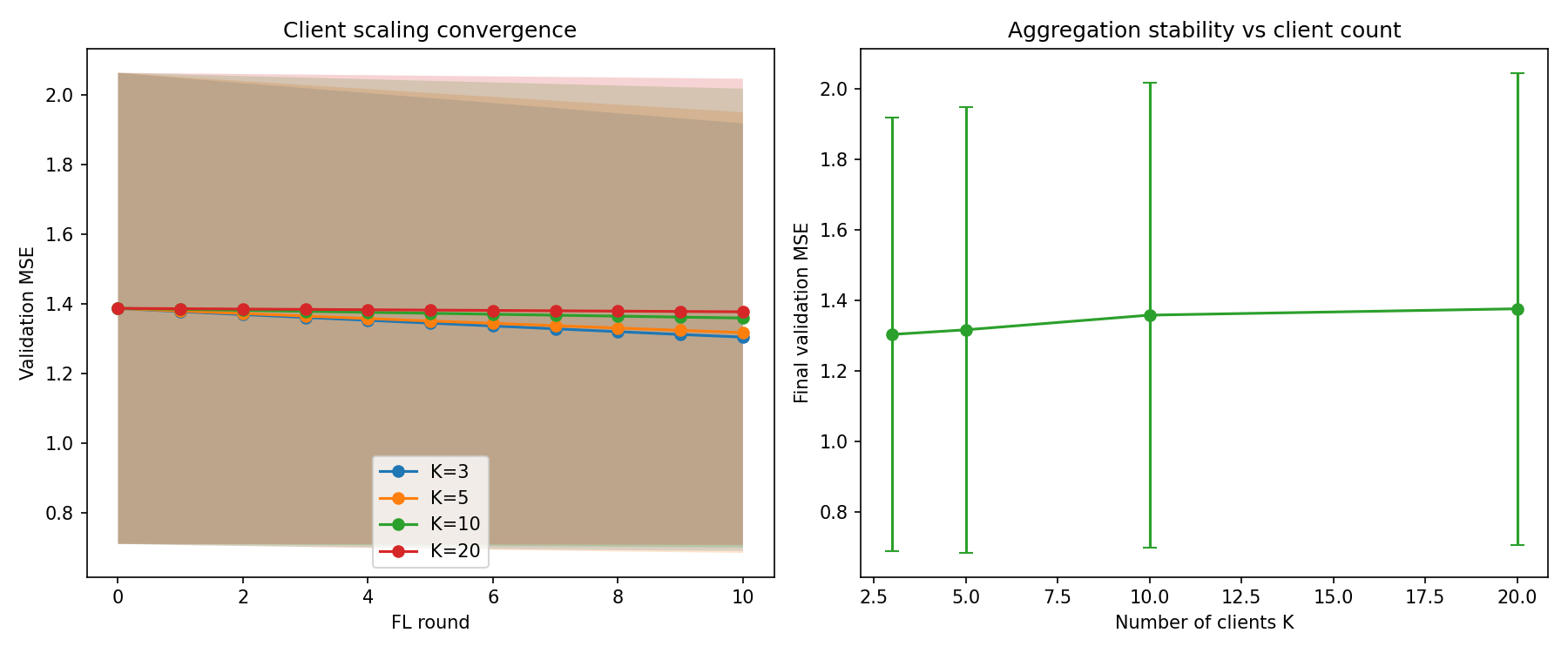}
        \caption{California Housing, $K\in\{3,5,10,20\}$.}
    \end{subfigure}\hfill
    \begin{subfigure}[t]{0.48\textwidth}
        \centering
        \includegraphics[width=\linewidth]{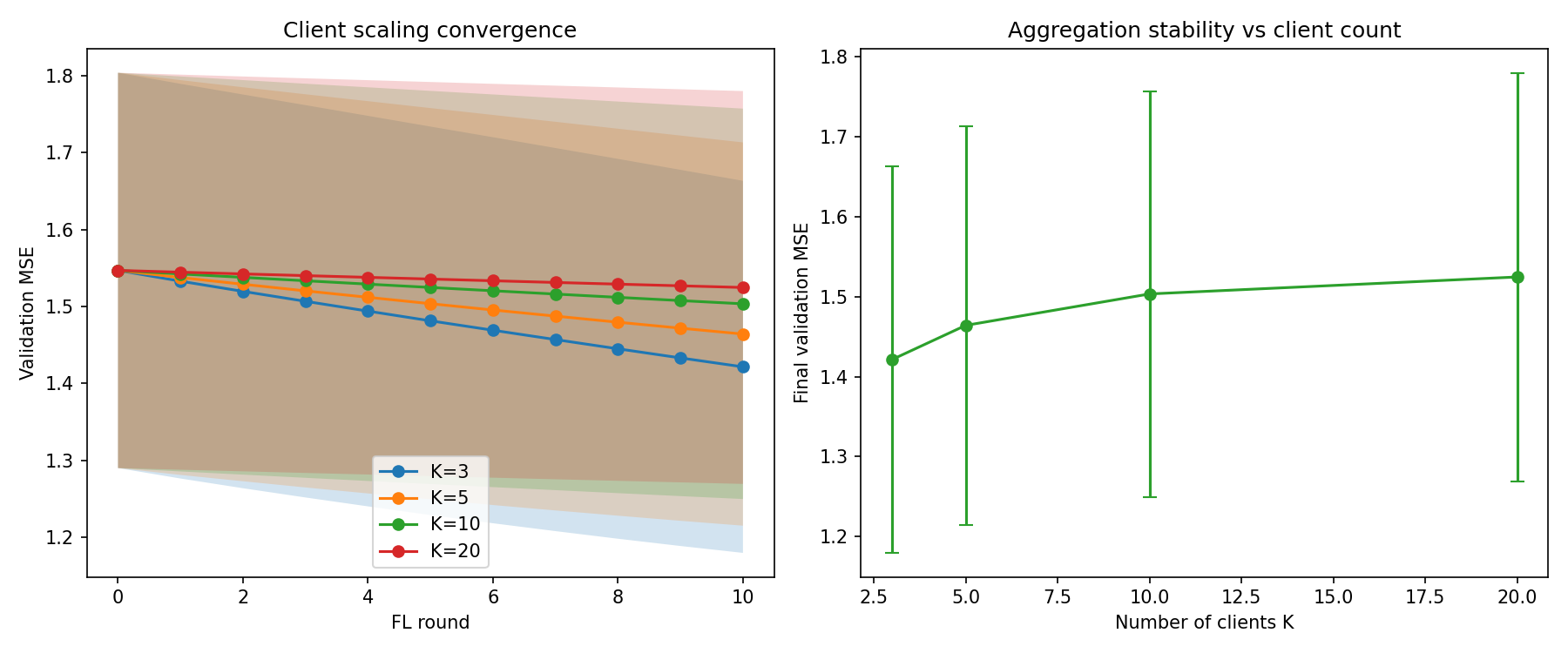}
        \caption{Wine Quality, $K\in\{3,5,10,20\}$.}
    \end{subfigure}
    \caption{E6: client scaling replicates on an unrelated dataset, with no
    aggregation blow-up at $20$ clients.}
    \label{fig:e6}
\end{figure}

\subsection{Hardware experiments}
\label{exp:ablation1}

\paragraph{Claim.} The encrypt $\to$ rotate $\to$ decrypt round trip preserves
single-qubit circuit behaviour on physical hardware, and the encryption layer's
own contribution is below the device error floor.
\paragraph{Procedure.} Sample rotation angles; for each trial run encrypted
state preparation, an encrypted $R_x(\theta)$, and decrypt-and-run; compute
circuit-level fidelity from measured $Z$-basis counts against the analytic
reference; separately measure a homomorphic-layer-only round trip to isolate
arithmetic noise. The matched unencrypted control repeats the identical seeded
sampling scheme and shot budget ($100$ angles, $4096$ shots) on the same device
with a calibration snapshot recorded alongside.
\paragraph{Assumptions.} The $Z$-basis estimator is a lower bound on state
agreement. Simulator and hardware tracks use different shot budgets and are
reported separately.
\paragraph{Result.} Encrypted $F=0.9918$ on \texttt{ibm\_fez}; matched
unencrypted control $F=0.99957$ (s.d.\ $0.00107$, min $0.99312$); simulator
$F=0.9989$; homomorphic-layer-only angle MAE $1.55\times10^{-4}$\,rad. The
difference between encrypted and control, $0.0078$, is comparable to the
device's median readout error at the time ($0.0083$; two-qubit gate error
$0.0029$).

\begin{figure}[H]
    \centering
    \begin{minipage}[b]{0.48\linewidth}
        \centering
        \includegraphics[width=\linewidth]{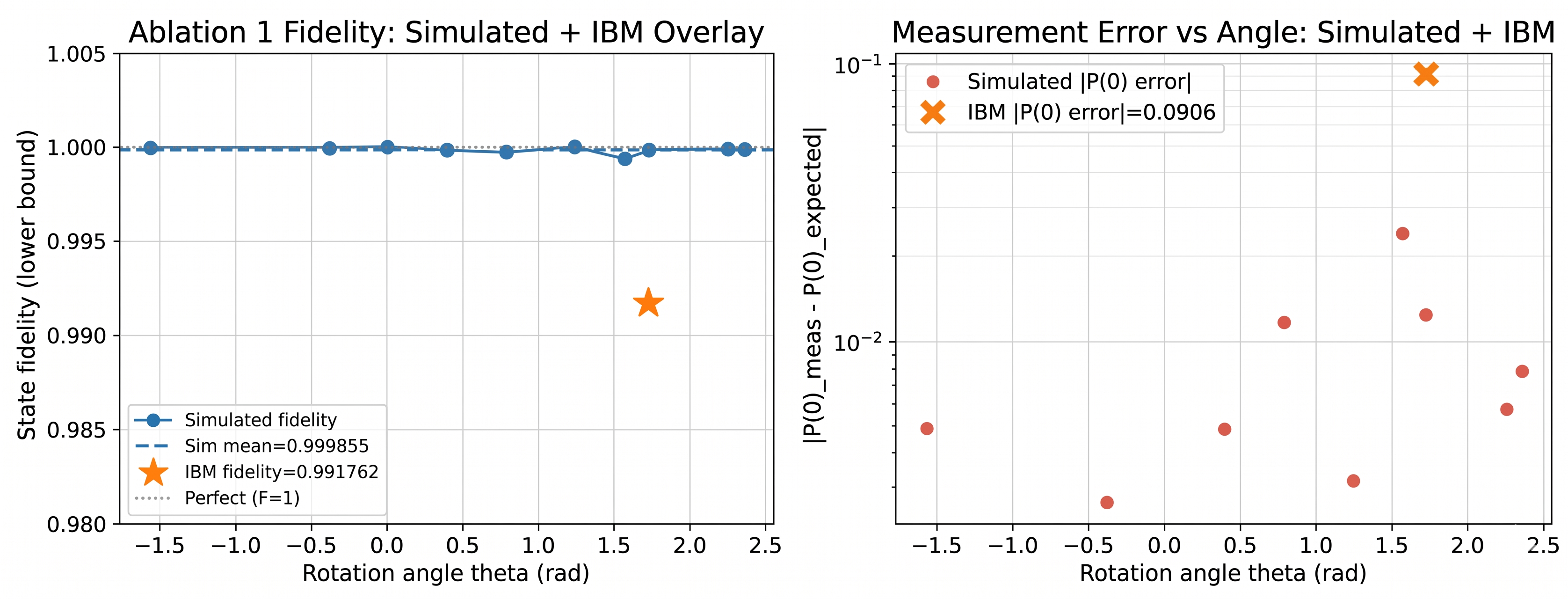}
        \caption{Round-trip state fidelity for $R_x(\theta)\ket{0}$ under
        encrypt $\to$ rotate $\to$ decrypt, simulator and \texttt{ibm\_fez}.}
        \label{fig:fid-primary}
    \end{minipage}\hfill
    \begin{minipage}[b]{0.48\linewidth}
        \centering
        \includegraphics[width=\linewidth]{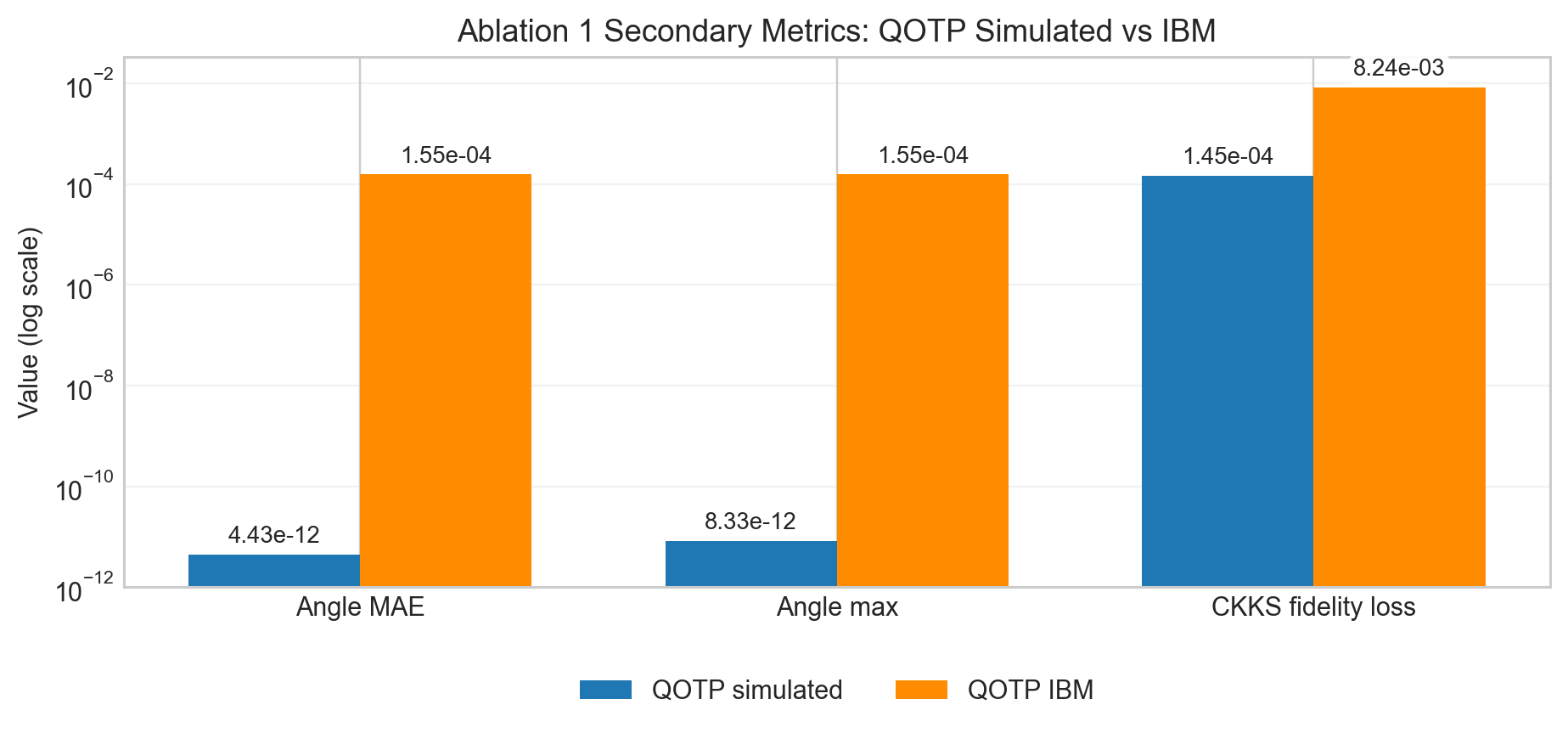}
        \caption{Homomorphic-layer angle-recovery error for the same round trip;
        an order of magnitude below the device gate-error floor.}
        \label{fig:fid-secondary}
    \end{minipage}
\end{figure}

\subsection{Supporting ablations}

\paragraph{Ablation 2: single-client utility preservation.}\label{exp:ablation2}
Round-tripping a trained client's angles through encryption and re-evaluating
leaves validation MSE unchanged in the reported run -- the single-client
special case of E1.

\paragraph{Ablation 3: multi-round error accumulation and runtime.}\label{exp:ablation3}
Repeated round trips across local training rounds show no observable angle
drift; the encryption stage dominates per-round runtime
($\approx167$\,s against $\approx76$\,s of local training).

\paragraph{Ablation 4: homomorphic aggregation error.}\label{exp:ablation4}
Superseded, and made precise, by E2: see Table~\ref{tab:exact}.

\paragraph{Ablation 5: encryption scope versus utility.}\label{exp:ablation5}
Three matched conditions (plaintext, quaternion-only, full homomorphic) under
one federated scaffold. Single-seed differences between conditions fall well
inside the between-seed spread measured in E1 and are therefore not
interpretable as treatment effects; the paired design of E1 supersedes this
comparison.

\paragraph{Ablation 6: precision sweep.}\label{exp:ablation6}
Error decreases sharply with precision and runtime increases; the two-backend
version is E2.

\paragraph{Ablation 7: secure-refresh frequency.}\label{exp:ablation7}
Sweeping refresh interval $k\in\{1,2,5,10\}$ against a no-refresh baseline, the
no-refresh configuration is best in the current implementation
($\mathrm{MAE}=1.24\times10^{-4}$\,rad against $2.7$--$4.3$\,rad for periodic
cadences), which is why the protocol aggregates in a single pass. This
characterises the present secure-runtime implementation, not all homomorphic
backends.

\section{Cost Models}
\label{app:cost-models}

\subsection{Hardware profile}\label{exp:costmodel1}
Calibration timings for the deployed workflow:
$T_{\text{rot,qotp}}\approx8856.5$\,ms against
$T_{\text{rot,pauli}}\approx11844.3$\,ms;
$T_{\text{to,qotp}}\approx5487.5$\,ms against
$T_{\text{to,pauli}}\approx9221.1$\,ms; analytic crossover
$p^\star\approx0.970$. Descriptive profiling rather than a controlled sweep.

\subsection{Interaction versus bandwidth}\label{exp:costmodel5}
Protocol accounting over $50$ federated rounds and three clients gives $50$
sequential exchanges in total for this protocol against $1{,}850$ (Pauli-OTP)
and $3{,}650$ (blind delegation), at $6.33$\,GB total traffic against
$\approx156$\,MB. Table~\ref{tab:crossover} converts both into wall-clock terms.
Fixed ciphertext size and message accounting are assumed; totals are model-based
rather than network traces.

\begin{table}[H]
\caption{Where non-interactivity pays. Left: time to move this protocol's
$45.3$\,MB per federated round. Right: the latency floor of interactive
baselines at $37$--$73$ sequential exchanges, before any payload moves.}
\label{tab:crossover}
\centering
\small
\begin{tabular}{@{}lcc@{}}
\toprule
\textbf{Link} & \textbf{This protocol (payload time)} & \textbf{Interactive baselines (latency floor)} \\
\midrule
$1$\,Gbps, $50$\,ms RTT   & $0.36$\,s & $1.9$--$3.7$\,s \\
$100$\,Mbps, $50$\,ms RTT & $3.6$\,s  & $1.9$--$3.7$\,s \\
$10$\,Mbps, $50$\,ms RTT  & $36$\,s   & $1.9$--$3.7$\,s \\
any link, $1$\,ms RTT     & as above  & $0.04$--$0.07$\,s \\
\bottomrule
\end{tabular}
\end{table}

\subsection{Solovay--Kitaev comparison}\label{exp:costmodel6}
Per-rotation compute reduction against Solovay--Kitaev decomposition, from gate
counts at matched target precision: $10\times$ at $\epsilon=10^{-2}$,
$20\times$ at $10^{-4}$, $30\times$ at $10^{-6}$. This is an accounting estimate
over gate counts, not an end-to-end wall-clock measurement, and should not be
read as one.

\subsection{Pure-quantum diagnostic}\label{purequantum}
A four-qubit data-reuploading circuit with only minimal classical output
calibration~\citep{Ballester2025} fails to converge under the same encrypted
federated pipeline: validation MSE rises monotonically from $1.24$ to $1.56$
across five rounds. Together with Ablation~7 this makes the hybrid co-design an
architectural necessity rather than a convenience: classical components carry
the gradient signal through landscapes a small variational circuit cannot
navigate~\citep{mcclean2018barren}, and single-pass aggregation avoids
mid-training key refreshes that current secure runtimes handle poorly.

\section{Scale Tests}
\label{app:scale-models}

The scale suite (client scaling~\ref{exp:scale2}, horizon~\ref{exp:scale3},
bandwidth~\ref{exp:scale4}) is reported at full breadth in E6; this section
records the original single-dataset sweeps for completeness.

\paragraph{Multi-seed variance.}\label{exp:scale1}
Five seeds, ten rounds, three clients. This sweep runs the two arms sequentially
over shared loaders, so the arms differ in initialisation and batch order as
well as in encryption; the resulting apparent gap is initialisation variance and
is superseded by the paired design of E1, which measures the encryption effect
at $+9\times10^{-6}$ MSE against a between-seed spread of $\approx0.72$.

\paragraph{Client scaling.}\label{exp:scale2}
Final MSE grows mildly and monotonically with participant count; tail standard
deviation \emph{improves} at larger $K$ ($0.0178\to0.0061$), i.e.\ late-round
behaviour is smoother with more clients. Extended to $K=20$ on two datasets in
Table~\ref{tab:scale}.

\paragraph{Extended horizon.}\label{exp:scale3}
Final MSE decreases monotonically with horizon and cross-seed spread contracts,
on both datasets (Table~\ref{tab:scale}).

\paragraph{Bandwidth sensitivity.}\label{exp:scale4}
Monotonic trade of angle error against traffic (Table~\ref{tab:bandwidth}).
Communication totals are protocol byte accounting, not network traces.

\section{Implementation and Security Profiles}
\label{app:security}

\paragraph{Cryptographic backends.}
The executed stack implements the GSW-style LWE constructions of
\citet[Schemes 5.1 and 5.2]{Ma_2022} over fixed-point bit-level arithmetic on
the quaternion components $t=(t_1,t_2,t_3,t_4)$ and Euler angles
$(\alpha,\beta,\gamma)$, combined with pairwise Diffie--Hellman
masking~\citep{diffie1976new} and information-theoretic mask cancellation for
the secure-aggregation harness. A parallel CKKS path over
OpenFHE~\citep{badawi2022openfhe} (ring dimension $2^{16}$, multiplicative depth
$40$, $50$-bit scaling modulus, \texttt{FLEXIBLEAUTO}, \texttt{HYBRID} key
switching) is used for the backend cross-check of E2. Security labelling should
therefore track the path: the executed stack is LWE-based, the CKKS
instantiation RLWE-based~\citep{regev2009lattices,lyubashevsky2010ideal}.

\begin{table}[H]
\caption{Threat model. Decryption capability is shared across clients and
withheld from the server (the aggregation object exposes only add, scale and
serialise). Threshold key generation is the drop-in hardening for the peer row;
see Appendix~\ref{app:security}.}
\label{tab:threat}
\centering
\small
\begin{tabular}{@{}p{2.5cm}p{4.3cm}p{5.6cm}@{}}
\toprule
\textbf{Party} & \textbf{Observes} & \textbf{Does not obtain} \\
\midrule
Cloud QPU & Quaternion-masked states & The underlying state. Masking is
information-theoretic with per-client, per-use keys removed by their owner. \\
Aggregation server & Ciphertexts, message sizes, client weights $n_k/n$ & Any
plaintext. Decryption raises an error by construction. \\
Peer client & The decrypted global aggregate & Nothing further \emph{by
protocol}; under the shared-key model a peer holding another's ciphertext could
decrypt it, the standard shared-key FHE-FedAvg assumption. \\
Network adversary & Authenticated ciphertext traffic & Plaintexts; undetected
tampering. \\
\bottomrule
\end{tabular}
\end{table}

\paragraph{Parameter profiles, and what the reported runs are entitled to claim.}
The default profile is a structural-verification setting ($n=1$, $q=256$, $m=8$,
$\beta_{\text{init}}=4$), which satisfies the noise bound
$\beta\geqslant2\sqrt{n}$ of~\citet{Ma_2022} and is chosen for auditability, not for
security level. Production profiles ($n=64$, $q=2^{32}$,
$\beta_{\text{init}}=16$; and $n=256$, $\beta_{\text{init}}=32$) are available.
\emph{No security property is claimed for the reported runs}: they establish
protocol structure, exactness and depth, and the security argument rests on the
scheme, not on the profile used for verification.

\paragraph{Angle extraction and the entangler path.}
Euler-angle extraction from $\Enc(t)$ uses CORDIC vectoring in fixed point,
fully oblivious (zero scalar decryptions), with no singularity at the origin and
precision controlled by $\max(12,\min(24,\text{bits}+2))$ iterations. The
encrypted-rotation pipeline is a levelled circuit whose level budget is
$3kL+1$ for $k$ precision bits and $L$ circuit levels ($37$ levels at the
default $k=12$, $L=1$); level promotion refreshes ciphertexts as needed and no
periodic carry bootstrapping is required. The Boolean sub-layer is where the
$\approx1.5$\,h per encrypted CNOT of \S\ref{sec:exp-entangler} is spent; a
Boolean-native backend~\citep{chillotti2020tfhe} is the natural remedy.

\begin{figure}[H]
  \centering
  \includegraphics[width=\linewidth]{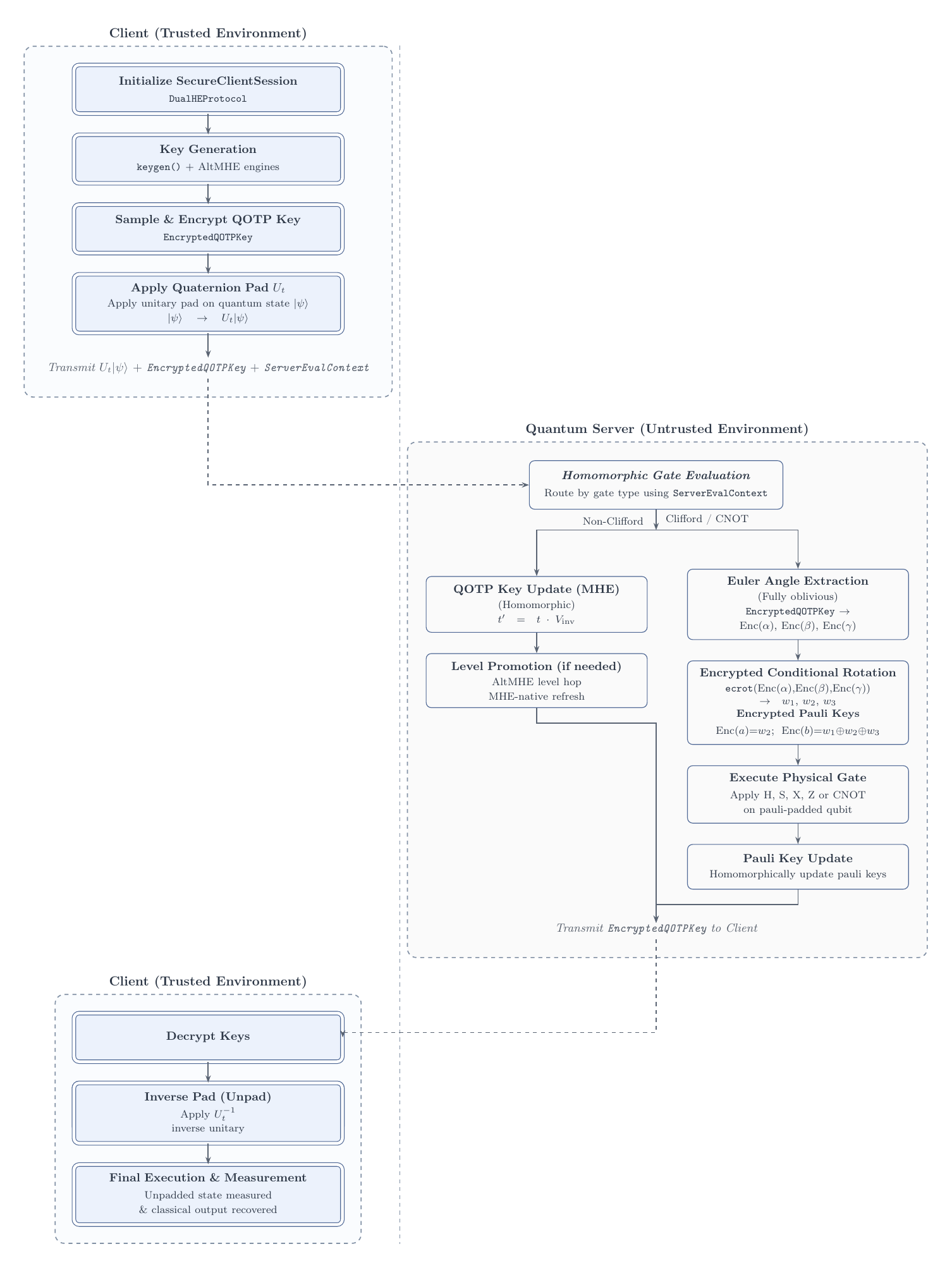}
  \caption{End-to-end implementation flow for the quaternion track.}
  \label{fig:flow}
\end{figure}

\section{Extended Discussion and Limitations}
\label{app:discussion-limitations}

\subsection{Bandwidth is the new bottleneck}

The dominant cost moves from sequential round-trips to per-round bandwidth:
$6.33$\,GB over a $50$-round run, driven by $22.67$\,MB ciphertexts per client
per round, with single-round encryption ($\approx167$\,s) exceeding local
plaintext training ($\approx76$\,s). The crossover analysis of
\S\ref{sec:comm} makes the boundary explicit rather than rhetorical: the trade
pays above roughly $100$\,Mbps at wide-area latency and does not pay below it or
on low-latency links. Precision recovers about a third of the traffic at no
measured utility cost (\S\ref{sec:exp-parity}), and ciphertext compression is
the obvious next lever.

\subsection{Security, in two layers}

The quantum layer provides information-theoretic security at the state level;
the classical layer provides computational security under (Ring-)LWE. The
system-level guarantee is bounded by the weaker link, so a future
polynomial-time attack on the lattice assumption would compromise aggregation
even though quantum ciphertexts remain opaque; practitioners should size
parameters to their threat model. Decryption capability is shared across
clients, which is the standard shared-key homomorphic-FedAvg assumption and the
one place where a peer is more powerful than the protocol's spirit suggests;
threshold key generation~\citep{mouchet2021multiparty} removes this at the cost
of an additional key-setup phase.

\subsection{Robustness against malicious clients}

A blind server cannot inspect updates, so poisoned updates are undetectable by
construction -- the standard Byzantine pathology, sharpened by encryption.
Robust aggregation compatible with the homomorphic layer, along the lines
of~\citet{ElMaouaki2025RobQFL}, is the natural next component. Until then the
protocol suits institutionally vetted federations where client identity is
verified.

\subsection{Implementation scope}

The core protocol path is implemented and validated end to end; full
application-layer integration into a production federated
stack~\citep{beutel2020flower} is incomplete, and several planned cost-model
studies (evaluation time versus CNOT fraction, depth scaling at fixed CNOT
fraction, gate ordering) were not executed. Hardware validation is constrained
by quantum-cloud budgets: at prevailing runtime pricing~\citep{IBMQuantumProducts}
multi-seed hardware sweeps at the scale conventional in classical ML are
structurally out of reach for most groups. We record this as a limitation of the
evidence, and separately as an open question for the community about whether
empirical norms calibrated for classical compute are the right ones to apply to
hardware-constrained research.

\section{Broader Impact}
\label{app:broader-impact}

\paragraph{Privacy.} The construction strengthens confidentiality in federated
training: even an adversary who later breaks the classical layer obtains no
transcript of the delegated quantum computation, which existed only as transient
mixed states. Organisations should nonetheless evaluate whether the residual
lattice-based classical-layer security matches their threat model; the classical
scheme is replaceable without touching the reduction.

\paragraph{Auditability.} Blindness prevents compliance auditing as well as
surveillance. In jurisdictions requiring model inspection or data-access rights,
the protocol's opacity may conflict with existing obligations, and
selective-disclosure mechanisms compatible with the security model are a
worthwhile line of work.

\paragraph{Adversarial robustness.} Under the honest-but-curious model a
minority of malicious clients could degrade the global model with no audit
signal. Deployment in adversarial environments should wait for
homomorphic-compatible robust aggregation.

\paragraph{Access and equity.} Physical validation required institutional
quantum-cloud and HPC access. As the approach matures, benefits risk accruing to
already well-resourced actors. The bandwidth profile compounds this: community
clinics, mobile devices and low-connectivity nodes -- exactly the settings where
federated privacy guarantees are most socially valuable -- are least able to
absorb the cost. Ciphertext compression and lightweight backends are equity
concerns as much as performance ones.

\paragraph{Honest framing of readiness.} The hybrid design is a necessity, not a
convenience: the pure-quantum diagnostic degrades monotonically, and periodic
key refresh in the current runtime is catastrophic relative to single-pass
aggregation. Readers should not extrapolate the reported numbers beyond the
configuration validated here.

\section{Use of Large Language Models}
\label{app:llm}

Large language models assisted with code development, figure generation, and polishing the manuscript's language. AI was also used for literature retrieval and discovery (e.g., to initially discover papers and place our claims in the broader context of the field). Finally, AI tools helped in generating the synthetic datasets used for the small-scale empirical evaluations, such as the synthetically scattered client rotation angles ($\sigma \in \{0.05, 0.2, 0.5, 1.0\}$ rad) used in the aggregation-geometry stress tests, and the seeded random angles generated for the hardware evaluations. All experimental design, final analysis, and scientific claims remain the authors' own work.

\end{document}